\documentclass[11pt]{article}

\usepackage{header}
\usepackage{macros}

\title{Coloring Fast with Broadcasts}
\date{}

\author{
Maxime Flin\\
\small Reykjavik University\\
\small \texttt{maximef@ru.is}\and
Mohsen Ghaffari\\
\small MIT\\
\small \texttt{ghaffari@mit.edu}\and
Magn\'us M. Halld\'orsson\\
\small Reykjavik University\\
\small \texttt{mmh@ru.is}\and
Fabian Kuhn\\
\small University of Freiburg\\
\small \texttt{kuhn@cs.uni-freiburg.de}\and
Alexandre Nolin\\
\small CISPA\\
\small \texttt{alexandre.nolin@cispa.de}
}

\begin{document}

\maketitle

\begin{abstract}
We present an $O(\log^3\log n)$-round distributed algorithm for the $(\Delta+1)$-coloring problem, where each node \emph{broadcasts} only one $O(\log n)$-bit message per round to its neighbors. Previously, the best such broadcast-based algorithm required $O(\log n)$ rounds. If $\Delta \in \Omega(\log^{3} n)$, our algorithm runs in $O(\log^* n)$ rounds. Our algorithm's round complexity matches state-of-the-art in the much more powerful \congest model [Halldórsson et al., STOC'21 \& PODC'22], where each node sends one different message to each of its neighbors, thus sending up to $\Theta(n\log n)$ bits per round. This is the best complexity known, even if message sizes are unbounded. 

Our algorithm is simple enough to be implemented in even weaker models: we can achieve the same $O(\log^3\log n)$ round complexity if each node reads its received messages in a streaming fashion, using only $O(\log^3 n)$-bit memory. Therefore, we hope that our algorithm opens the road for adopting the recent exciting progress on sublogarithmic-time distributed $(\Delta+1)$-coloring algorithms in a wider range of (theoretical or practical) settings. 
\end{abstract}

\pagenumbering{roman}
\thispagestyle{empty}
\newpage
\thispagestyle{empty}
\tableofcontents
\newpage
\pagenumbering{arabic}

\section{Introduction}

\noindent\textbf{The coloring problem and its distributed motivations.} Our focus is on $\Delta+1$-coloring: the problem of assigning one color from $\set{1,\ldots,\Delta+1}$ to each node, such that no two neighboring nodes have the same color. Here $\Delta$ denotes the maximum degree of the graph. Coloring plays a pivotal role in distributed systems, as a clean way to divide access to non-shareable resources, resolve contention, and break symmetries. For instance, it is particularly important in wireless networking, for frequency allocation or channel assignment. A characteristic of wireless communication is that nodes broadcast their messages (reception is constrained by interference from other broadcasts). 

\paragraph{Distributed models.} The coloring problem has been studied extensively in distributed computing~\cite{PanconesiS97,johansson99,SW10,FHK16,BEPSv3,HSS18,CLP20,GGR20,HKMT21,GK21,HKNT22,HNT22}. Indeed, this problem was the subject of the celebrated paper by Linial~\cite{linial92}, which introduced the \local model of distributed computing. In this model, $n$ processors form a graph $G=(V, E)$ where an edge exists only between processors that can communicate. The resulting graph is called the communication graph $G$ and is the one to be colored. Per round, each node can send one unbounded-size message to each of its neighbors. The variant where the message sent to each neighbor is bounded to $O(\log n)$ bits is known as the \congest model~\cite{peleg00}.

\paragraph{Distributed coloring.} Classic distributed algorithms for coloring~\cite{luby86, johansson99} achieved complexity $O(\log n)$ in the \congest model. There has been exciting recent progress on sublogarithmic time algorithms~\cite{BEPSv3,HSS18, CLP20,GGR20,HKMT21,GK21,HKNT22,HNT22}, and the state of the art round complexity is $O(\log^3\log n)$ rounds. This is also the best known in the more relaxed \local model, which allows unbounded message sizes. However, unlike the earlier algorithm of \cite{johansson99}, these faster algorithms make some nodes send one different message to each of their neighbors. Thus, each node may send up to $\Theta(n\log n)$ bits in one round. The research question at the core of this paper is to understand the extent to which one can compute a coloring fast if we constrain the set of outgoing messages.
Specifically,
\begin{quote}
    \emph{Can we compute a $(\Delta+1)$-coloring as fast as in the \congest model if, in each round, each node must transmit the same $O(\log n)$-bit message to all its neighbors?}
\end{quote}
To the best of our knowledge, with this restriction, the best round complexity known in general graphs remains the classic $O(\log n)$ bound~\cite{luby86, johansson99, BEPSv3}.
\subsection{Our Results}

We give a fast $\Delta+1$-coloring algorithm in the \emph{broadcast congest} model (or \bcongest) where, per round, each node \emph{broadcasts one $O(\log n)$-bit message} to all of its neighbors.

\vspace{1em}
\begin{restatable}{shadetheorem}{bcongestthm}
\label{thm:bcongest}
Let $G=(V,E)$ be any $n$-node graph with maximum degree at most $\Delta$. There is a distributed $O(\log^3\log n)$-round algorithm that $\Delta+1$-colors $G$ with high probability, where each node broadcasts one $O(\log n)$-bit message in each round. If $\Delta\in\Omega(\log^3 n)$, the algorithm runs in $O(\log^* n)$ rounds.
\end{restatable}

As a side remark, we note that the $O(\log n)$ complexity was the best bound known for general graphs even in the much more relaxed \emph{broadcast congested clique} model, in which each node can send a $O(\log n)$ bit message to all other nodes. To emphasize, in this model, the communication graph is a complete graph and every two nodes are neighbors. The coloring is still with respect to the input graph $G$. This model is also sometimes known as the \emph{shared blackboard model} with simultaneous messages and the \emph{distributed sketching} model~\cite{drucker2014power,assadi2020lower,AKZ22}. Our $O(\log^3\log n)$-round complexity improves nearly exponentially over existing algorithms in this model.\footnote{If we increase the size of the message sent by each node in this \BCC model from $O(\log n)$ to $O(\log^3 n)$ bits, then a celebrated work of Assadi, Khanna, and Chen~\cite{ACK19} provides a one round algorithm.}

\paragraph{Even more basic models?} The overarching goal in our work is not tied to any particular model. We would like to develop a distributed algorithm that assumes the least provided power from the theoretical model. The hope is that this makes the algorithm applicable in a wider range of (theoretical or practical) settings. To that end, we point out that our algorithm is basic enough to be implemented even with limited memory per node, with only small additional changes. Notice that a node may receive many messages from its neighbors, up to $\Omega(n\log n)$ bits overall in one round. In general, receiving so many bits would necessitate a significant memory for the node, and it also can complicate the task of simulating this algorithm in virtual graphs.~\footnote{For instance, consider a frequent scenario in distributed graph algorithms: a virtual graph is formed by contracting low-depth clusters of the network, each forming one node of the virtual graph. Two clusters are neighbors if they contain adjacent network nodes. Usually, the communications of each cluster should be sent along a low-depth tree that spans the nodes of the cluster. If all the $\Omega(n\log n)$ bits should be delivered to the cluster center, this can require $\Omega(n)$ rounds, even for low-depth clusters.}
We show that our algorithm can be adapted to work with the same round complexity when each node processes its incoming messages in a streaming fashion, using only $\poly(\log n)$ memory. We refer to this model as \bcstream. See \cref{sec:streaming} for a formal definition of the model.

\vspace{1em}
\begin{shadetheorem}
\label{thm:streaming}
There is a distributed $O(\log^3\log n)$-round algorithm in \bcongest for $\Delta+1$-coloring graphs with high probability, even if each node reads its received messages through a stream and only has $\poly(\log n)$ memory.
\end{shadetheorem}

\subsection{Technical Contributions}

\subsubsection{Previous Algorithms \& Challenges}
We summarize the key concepts in previous fast coloring algorithms and emphasize the parts that do not work in the \bcongest model. 

A basic primitive in randomized coloring algorithms is a \emph{random color trial}: each node selects a color from its \emph{palette} (its set of available colors) uniformly at random and keeps the color if none of its neighbors picked the same.
The (permanent) \emph{slack} of a node is the excess number of colors in its palette  compared to its degree. Sufficient slack speeds up coloring dramatically: each node can try \emph{multiple colors} in each round, resulting in a $O(\log^* n)$-round coloring algorithm called \multitrial \cite{SW10}. As a color requires up to $O(\log n)$ bits to describe, trying more than a constant number of them is infeasible with $O(\log n)$ bandwidth. A solution by \cite{HNT22} was to use pseudorandomness: say each $v$ tries a set of colors $X_v$, then $v$ broadcasts a hash function $h_v$ which each neighbor $u$ of $v$ uses to reply $h_v(X_u)$. A color that collides under $h_v$ with none of its neighbors is safe to adopt. However, this approach requires individual responses $h_v(X_u)$ from each neighbor $u$. Therefore it does not work with single-message broadcasts.
 
\begin{quote}
    \emph{Challenge 1: How can we perform \multitrial with $O(\log n)$-bit broadcasts? The previous approaches \cite{SW10, HNT22} require either large messages or individual responses.}
\end{quote}

Slack can be generated for nodes with a sparse neighborhood, i.e., with $\Omega(\Delta^2)$ missing edges. 
The more difficult task in distributed $\Delta+1$-coloring algorithms is to color the \emph{dense nodes}. They can be partitioned into
dense clusters called almost-cliques.
The second key concept for fast coloring is to \emph{synchronize} the colors tried within each almost-clique, in the following sense: the color suggested to each node should be random from the viewpoint of the nodes outside the almost-clique, but there should be no conflicts between nodes inside the almost clique.
The earlier version of synchronized color trial (SCT for short) involved gathering all the information of the almost-clique for centralized processing \cite{HSS18,CLP20}, requiring high bandwidth. 
A simpler form of SCT of \cite{HKNT22} has a leader node permute its own palette and distribute the colors to the other nodes of the almost-clique. 
This still requires different messages to be sent along the different edges from the leader, making it incompatible with \bcongest.
\begin{quote}
    \emph{Challenge 2: How can we synchronize color trials with $O(\log n)$-bit broadcasts? The previous approaches~\cite{HSS18,CLP20,HKNT22} require either centralization or a node sending up to $\Omega(\Delta)$ messages.}
\end{quote}

Finally, \multitrial requires $\ell=\Omega(\log^{1 + \Omega(1)} n)$ slack in order to fully color the graph with high probability. This is solved in \cite{HKNT22} by putting aside mutually non-adjacent sets of $\ell$ nodes in very dense cliques, to be colored at the very end. 
\cite{HKNT22} colors put-aside sets by gathering all their relevant information (list of uncolored neighbors and palette) and broadcasting the coloring from a leader node.

\begin{quote}
    \emph{Challenge 3: How can we color the put-aside sets with $O(\log n)$-bit broadcasts? The previous approach~\cite{HKNT22} does not work as they require full information gathering and dissemination.}
\end{quote}

Observe that Challenges 1 and 3 can easily be solved by increasing the bandwidth to a small $\poly(\log n)$. On the other hand, Challenge 2 seems to require greater effort to implement with the broadcast constraint, even with $\poly(\log n)$ bandwidth.

\subsubsection{Our Algorithm}
In this section, we give an overview of our solutions to each of the challenges described earlier.

\paragraph{Multi-Color Trial.}
A subset of a \emph{known} universe can be sampled pseudorandomly in \bcongest \cite{HN23}. The problem is that when \multitrial is applied after SCT, each node has a different palette, which is unknown to its neighbors. 
We solve this by \emph{reserving} a subset of the color space for use by \multitrial. Namely, each node $v$ reserves the subset $[x(v)] = \{1, 2, \ldots, x(v)\}$, where $x(v)$ is a function of $v$'s neighborhood density. Both slack generation and the synchronized color trial within $v$'s almost-clique are restricted to using colors outside $[x(v)]$. The key is then to show that: a) using the colors $[\Delta+1] \setminus [x(v)]$ suffices for these steps, and b) enough colors in $[x(v)]$ remain unused (by neighbors of $v$) for \multitrial to succeed.

\paragraph{Synchronized Color Trial.}
Our solution for the synchronized color trial of an almost-clique $K$ is to use the \emph{clique palette} of $K$: the set of colors not used by nodes in $K$. 
We randomly permute this set, in a distributed manner, and assign each color to a single uncolored node of $K$. This introduces two types of errors: a) not all nodes receive a color to try, and b) nodes can receive non-usable colors (as a node's neighbors outside of $K$ might already be using its assigned color). However, the errors are within acceptable bounds, and we are still able to show that after SCT, each node has an uncolored degree that is at most proportional to its slack, allowing for fast mop-up by \multitrial.

To learn the clique palette $\pal{K}$ in an almost-clique $K$, we randomly assign nodes of $K$ into groups such that: a) every node is adjacent to at least one node of each group, and b) each group is connected and has a low diameter. Each group is tasked with learning a part of the clique palette, which it teaches to the rest of the almost-clique $K$. 

We also randomly assign nodes into groups to randomly permute $K$. The random assignment roughly positions each node within the output permutation $\pi$. Each group, of much smaller size than $K$, then randomly permutes its members. The small size of each group, combined with relabeling its members with smaller {\ID}s, makes the description of a permutation of its members fit within small bandwidth.

\paragraph{Coloring Put-Aside Sets.}
The put-aside set $P_K$ of an almost-clique $K$ has no edges to the put-aside sets in other almost-cliques. As such, coloring $P_K$ can be done purely within $K$. Our algorithm first reduces the size of each $P_K$ to sublogarithmic. Then, it gathers information about what remains of each $P_K$.
One randomized color trial reduces $|P_K|$ by a constant factor with probability $1-e^{-\Theta(|P_K|)}$. We compress the equivalent of $O(\log\log n)$ iterations of this process into $O(1)$ rounds by sampling the colors of all iterations \emph{in advance} and sending them all at once.
To reach sublogarithmic size \emph{with high probability}, we run $O(\log\log n)$ independent iterations in parallel.
We avoid congestion issues by using few colors per iteration and by representing colors with few bits.

\subsection{Related Work}

\paragraph{Distributed $\Delta+1$-Coloring.}
The best round complexity of randomized \local $(\Delta+1)$-coloring, as a function of only the number of nodes $n$, progressed from $O(\log n)$ in the 80's \cite{luby86,alon86,johansson99}, through $O(\sqrt{\log n})$ \cite{HSS18}, to a recent $O(\log^3 \log n)$ \cite{CLP20}. 
The more recent work \cite{HSS18, CLP20} made heavy use of both the large bandwidth and the multiple-message transmission feature of the \local model. 
A crucial concept in these algorithms is \textit{shattering}. For coloring, shattering means coloring almost all the nodes such that each connected component of the set of nodes that remain uncolored has size at most $\poly(\log n)$. A similar concept was used originally by Beck \cite{beck1991algorithmic}. The idea was introduced to the distributed setting in \cite{BEPSv3}. The dominating factor in the time complexity is the deterministic complexity of solving (a variant of) the problem on polylogarithmic-sized problems. As there are now polylogarithmic-time algorithms for deterministic coloring \cite{RG20}, with the fastest being $O(\log^3 n)$ \cite{GK21}, the randomized complexity is currently $O(\log^3\log n)$ \cite{CLP20}.
An $O(\log^5\log n)$-round \congest algorithm was given in \cite{HKMT21}, improved to $O(\log^3 \log n)$ in \cite{HKNT22}. These algorithms still require transmitting different messages to all $\Omega(\Delta)$ neighbors of a node. 

Many distributed $(\Delta+1)$-coloring algorithms work immediately in \bcongest, including the folklore $O(\log n)$-round randomized algorithms \cite{johansson99} and the randomized part of \cite{BEPSv3}. The best deterministic algorithms known for small values of $\Delta$, with complexity $\tilde{O}(\sqrt{\Delta}) + O(\log^* n)$ \cite{FHK16,Barenboim16,MT20} use the full power of the \local model, however. 
 The $O(\log^3 n)$-round deterministic algorithm of \cite{GK21} also works in \congest, but it is sensitive to the palette size. When $\Delta \le \poly(\log n)$, \cite{GK21} with the shattering of \cite{BEPSv3} colors in $O(\log^3\log n)$. Otherwise, if $\Delta \gg \poly(\log n)$, dependency on the palette size can be resolved by relabeling the palette, using network decomposition \cite{GGR20}, as shown for coloring in \cite{HKMT21}. Hence, there is a $O(\log \Delta + \poly(\log\log n))$-round \bcongest algorithm for $(\Delta+1)$-coloring.

 While most known algorithms which work in \bcongest were published as \congest algorithms, without making explicit that they also work with broadcast communication, explicit mentions of \bcongest are becoming more and more frequent in recent years~\cite{CM19,PP19,FV22}.

\paragraph{Distributed Sketching and Broadcast Congested Clique.}
The palette sparsification theorem of \cite{ACK19} shows that even if each node uniformly samples $O(\log n)$ colors, the graph can still be $\Delta+1$-colored while restricting each node to use only a sampled color. 
This has led to a (one-pass) streaming algorithm for $\Delta+1$-coloring using $O(n\poly(\log n))$ space.
It was recently shown that the actual coloring can also be computed distributively, in $O(\log^2 \Delta + \log^3 \log n)$ rounds of \CONGEST \cite{FGHKN22}.
We utilize several technical lemmas from the work of \cite{FGHKN22}, while the actual results are almost completely unrelated.

Palette sparsification is a one round/pass form of \emph{distributed sketching} (or shared blackboard), a technique of considerable current interest \cite{AGM12b,AKM22,AKZ22}. 
The nomenclature that is closer to our setting is the \emph{broadcast congested clique} \cite{drucker2014power,JN18,BMRT20}.
Whereas there are no non-trivial lower bounds in the Congested Clique model for problems related to coloring, there is a recent $\Omega(\log\log n)$-round lower bound for the Maximal Independent Set problem in the broadcast congested clique \cite{AKZ22}.

\subsection{Organization of the Paper}

After preliminary definitions and results in \cref{sec:prelim}, we formally describe our algorithm in \cref{sec:alg} and give a proof of \cref{thm:bcongest}. \cref{sec:implem-sct} details the \bcongest implementation of the synchronized color trial. We explain how to modify our algorithm for the \bcstream model in \cref{sec:streaming}.

\section{Preliminaries}
\label{sec:prelim}

\paragraph{Notation.}
For any integer $k \ge 1$, we denote the set $\set{1,2,\ldots, k}$ by $[k]$. For any tuple $(x_1, x_2, \ldots, x_k)$, we shall write $x_{\le i}$ for $(x_1, \ldots, x_i)$. Likewise, let $x_{<i}=(x_1, \ldots, x_{i-1})$.

The communication network is $G=(V,E)$, we denote by $n=|V|$ its number of vertices, for each $v\in V$ we call $d(v)$ its degree and $\Delta$ the maximum degree of $G$. For a vertex $v\in V$, we denote by $N_G(v)=\set{u\in V, uv\in E}$ its neighbors in $G$. We assume nodes have $O(\log n)$-bit unique identifiers named $\ID(v)$. In the \bcongest model, nodes of $G$ communicate by broadcasting $O(\log n)$-bit messages in synchronous rounds.

A partial coloring is a function $\col:V\to [\Delta+1]\cup\set{\bot}$ such that for any edge $uv\in E$, its endpoints receive different colors $\col(u)\neq\col(v)$ unless $\col(v)$ or $\col(u)$ is $\bot$ -- which stands for ``not colored". With respect to any partial coloring \col, we shall write $\hatd(v)$ for the \emph{uncolored} degree of $v$, i.e., its number of uncolored neighbors with respect to \col. More generally, for any $S\subseteq V$, we write $\h{S}$ to denote the set of uncolored nodes in $S$ (with respect to a partial coloring). Our algorithm computes a monotone sequence of coloring, that is, once we fix $\col(v)$, it never changes.

When we say an event happens \emph{with high probability}, or w.h.p.\ for short, we mean with probability $1-n^{-c}$ for any suitably large constant $c > 0$. We implicitly choose the constant $c$ large enough to union bound over polynomially many events.

\subsection{Sparse-Dense Decomposition}

The sparsity counts the number of missing edges in the neighborhood of a node, with the important detail that if a node has degree less than $\Delta$, each ``missing'' neighbor counts as $\Delta$ missing edges.

\begin{definition}[Sparsity]
\label{def:sparsity}
The \emph{sparsity} $\zeta_v$ of $v\in V$ is \[\zeta_v := \frac{1}{\Delta}\parens*{\binom{\Delta}{2} - m(N(v))}\ ,\]
where $m(N(v))$ is the number of edges induced by $N(v)$.
Node $v$ is $\zeta$-sparse if $\zeta_v \ge \zeta$ and $\zeta$-dense if $\zeta_v \le \zeta$.
\end{definition}

We decompose the graph between locally sparse nodes and dense clusters called \emph{almost-cliques}. Almost-cliques can be thought of as graphs that are $\epsilon$-close to $\Delta$-cliques, in a property-testing meaning. Such decomposition is ubiquitous in randomized coloring \cite{Reed98,HSS18,ACK19,CLP20,AA20,HKMT21}.

\begin{definition}
\label{def:almost-clique}
For $\epsilon \in (0, 1/3)$, an $\epsilon$-almost-clique decomposition is a partition of $V(G)$ in sets $\Vsparse, K_1, \ldots, K_k$ such that
\begin{enumerate}
\item nodes in $\Vsparse$ are $\Omega(\epsilon^2\Delta)$ sparse,
\item for all $i\in[k]$, almost-clique $K_i$ satisfies:
\begin{enumerate}
\item\label[part]{part:clique-size} $|K_i| \le (1+\epsilon)\Delta$,
\item\label[part]{part:anti-deg} $|N(v)\cap K_i|\ge (1-\epsilon)\Delta$ for all $v\in K_i$, and
\item\label[part]{part:deg-inside} $|N(v)\cap K_i| \le (1-\epsilon/2)\Delta$ for all $v\notin K_i$. 
\end{enumerate}
\end{enumerate}
\end{definition}

\begin{definition}[External and Anti-Degrees]
For a node $v\in K$ and some almost-clique $K$. We call $e_v = |N(v)\setminus K|$ its \emph{external degree} and $a_v=|K\setminus N(v)|$ its anti-degree. 
We shall denote by $\avgext_K=\sum_{v\in C} e_v/|K|$ the average external degree and $\avganti_K=\sum_{v\in K}a_v/|K|$ the average anti-degree.
\end{definition}

\Cref{part:deg-inside} is not typically included in prior work (e.g., \cite{ACK19,HKMT21}). It was used recently in \cite{AKM22,FHM23}. We use it solely to prove \cref{lem:slack-ext}. We call \emph{anti-edge} a missing edge between two nodes, i.e., an edge in the complement graph.

\begin{lemma}
\label{lem:slack-ext}
Let $K$ be any almost-clique. Every $v\in K$ is $(\epsilon/2 \cdot e_v)$-sparse. 
\end{lemma}

\begin{proof}
Fix $v\in K$. We count the number of anti-edges in $(N(v)\cap K) \times (N(v)\setminus K)$.
Let $u\in N(v)\setminus K$ be an external neighbor of $v$. By \cref{part:deg-inside} of \cref{def:almost-clique}, vertex $v$ can have at most $(1-\epsilon/2)\Delta$ neighbors in $K$. Moreover, $v$ has at least $(1-\epsilon)\Delta$ neighbors in $K$ (by \cref{part:anti-deg}). Hence, there are at least $\epsilon\Delta/2$ anti-edges between $u$ and $N(v)\cap K$. Overall, the number of anti-edges between external and internal neighbors is at least $e_v\cdot\epsilon\Delta/2$.
\end{proof}

The first \congest algorithm to compute almost-clique decompositions in $O(1)$ rounds (when $\Delta \in \Omega(\log^2 n)$) was given by \cite{HKMT21}. It was then improved by \cite{HNT22} to arbitrary $\Delta$ in the \congest model. \cite{FGHKN22} gives a simpler implementation of \cite{HNT22} that works in \bcongest and \bcstream. 

\begin{restatable}[\cite{FGHKN22}]{lemma}{ACDbcongest}
\label{lem:acd-bcongest}
For any $\epsilon \in (0,1/20)$, there exists an algorithm computing an $\epsilon$-almost-clique decomposition in $O(\eps^{-4})$ rounds of $\bcongest$ with high probability.
\end{restatable}

\paragraph{Colorful Matching.} 
In a $\Delta+1$-clique, the colors used in the clique are exactly the colors used in the neighborhood of each node. An almost-clique can have size larger than $\Delta+1$. Thus,
an almost-clique with uncolored nodes might actually have an empty clique palette.
To solve this issue, \cite{ACK19} introduced the idea of colorful matching.

\begin{definition}[Colorful Matching]
\label{def:colorful-matching}
A \emph{colorful} matching in a clique $K$ (with respect to a partial coloring $\col$) is a matching of anti-edges in $K$ (edges in the complement graph) such that 1) endpoints of each anti-edge receive the same color, and 2) each anti-edge has a different color.
\end{definition}

Intuitively, if one contracts anti-edges of the colorful matching, one reduces the size of the almost-clique while maintaining a proper coloring. If the matching is large enough, the number of unused colors in $K$ is greater than the number of uncolored nodes.

\begin{definition}[Clique Palette]
\label{def:clique-palette}
For each $K$, let the \emph{clique palette} $\pal{K}=[\Delta+1]\setminus\col(K)$ be the set of colors not used in $K$.
\end{definition}

\begin{claim}
\label{claim:clique-palette}
Let $K$ be an almost-clique and $M$ a colorful matching in $K$. Then, for all $v\in K$
\[ |\pal{K}| \ge |\hK|+1 +e_v - a_v + |M|\ . \]
\end{claim}

\begin{proof}
The clique palette loses at most one color per colored node but saves one for each anti-edge in the colorful matching; hence, $|\pal{K}|\ge \Delta+1-(|K|-|\hK|)+|M|$. On the other hand, observe that $\Delta \ge |N(v)\cap K| + e_v$ and $|K| = |N(v)\cap K|+a_v$. The claim follows.
\end{proof}

By computing a matching of size $\Theta(\avganti_K)$, the clique palette always contains colors for each node in $\hK$. Computing a colorful matching of size $\Theta(\avganti_K)$ can be done in $O(1)$ rounds as the clique contains $\Theta(\avganti_K\Delta)$ anti-edges and by trying colors, we expect $\Theta(\avganti_K)$ edges to join the matching. 
A minor difference between our setting and the one of \cite{FGHKN22} is that when they compute the colorful matching, almost-cliques are fully uncolored. On the contrary, our algorithm colors a constant fraction of each almost-clique to produce slack (\cref{lem:slackgeneration}). In \cref{sec:appendix-colorful-matching}, we show that we loose only small fraction of the anti-edges in the clique when doing so; hence, that it does not impede the colorful matching algorithm.

\begin{restatable}[\cite{FGHKN22}]{lemma}{colorfulmatchingtheorem}
\label{thm:colorful-matching}
Let $\beta < 1/(18\epsilon)$ be a constant. There exists a $O(\beta)$-round algorithm called \matching that computes a colorful matching of size $\beta\cdot\avganti_K$ with probability $1-n^{-\Theta(C)}$ in every clique $K$ with $\avganti_K \ge C\log n$. Furthermore, at most $2\beta\cdot\avganti_K$ nodes are colored in each almost-clique during this step.
\end{restatable}

\subsection{Distributed Coloring with Slack}

\begin{definition}[Palette]
The palette $\pal{v}$ of node $v$, with respect to a partial coloring, is the set of colors not used by its neighbors.
\end{definition}

\begin{definition}[Slack]
\label{def:slack}
The \emph{slack} $s_H(v)$ of a node $v$ in a subgraph $H$ is the difference between the size of its palette and its uncolored degree in this graph: $s_H(v)=|\pal{v}| - \hatd_H(v)$. When $H$ is clear from context, we simply write $s(v)$.
\end{definition}

There are three ways a node can receive slack: if it has a small degree originally, if two neighbors adopt the same color, or if an uncolored neighbor is inactive (does not belong to $H$). We consider the first two types of slack \emph{permanent} because a node never increases its degree, and nodes never change their adopted color. On the other hand, the last type of slack is \emph{temporary}: if some previously inactive neighbors become active, the node will lose the slack that those inactive neighbors were providing before.
Elkin, Pettie, and Su \cite{EPS15} observed that by trying random colors, nodes would receive slack proportional to their sparsity. 

\begin{lemma}[Slack Generation, {\cite[Lemma 3.1]{EPS15}}]
\label{lem:slackgeneration}
Let $v$ be a $\zeta$-sparse node for some $\zeta$. Suppose each node of $G$ independently decides w.p. $\ps=1/200$ to try a uniform color in $[\Delta+1]$. Then, w.p.\ $1-e^{-\Theta(\zeta)}$, $v$ has slack $s(v) \ge \gamma \cdot \zeta$ where $\gamma > 0$ is a (small) universal constant.
\end{lemma}

\paragraph{Trying Colors From Lists.}
When we say a node \emph{tries a random color}, we mean that it broadcasts a color uniformly sampled from some set (usually from its palette) and \emph{adopts} the color if none of its neighbors with smaller \ID tried the same color. 
It is known that nodes with $\Omega(\log n)$ uncolored neighbors see a constant fraction of them get colored when they try random colors, w.h.p.\ {\cite{BEPSv3}}

\begin{restatable}{lemma}{RCTLemma}
\label{claim:rct-less-colors}
Let $H$ be a vertex-induced subgraph and $L(v)\subseteq \pal{v}$ for each $v$. Suppose there exists a globally known constant $\alpha > 0$ such that every uncolored $v$ satisfies $|L(v)| \ge \alpha\cdot\hatd(v) \ge C\log n$. If nodes independently call \trycolor w.p. $\pt=\alpha/3$ \emph{and samples a uniform color in $L(v)$}, then, w.p. $1-n^{-\Theta(C)}$, the uncolored degree of every node has decreased by a factor $2/3$.
\end{restatable}

Trying multiple colors to take advantage of extra colors (i.e., slack) was proposed originally by \cite{SW10}. It is a key component of all recent fast randomized coloring algorithms \cite{CLP20,HKNT22,HNT22}. A small tweak suffices to bring the technique to \bcongest. 

\begin{restatable}[Multi-Color Trial, \cite{HN23,HKNT22}]{lemma}{MCTLemma}
\label{lem:mct-bcongest}
Let $H$ be a vertex-induced subgraph of $G$. Suppose that for each $v\in H$, there is a $L(v)$ list of colors satisfying
\begin{enumerate}
    \item\label[part]{part:list-known} $L(v)$ is known by each $u\in N_H(v)$,
    \item\label[part]{part:list-dense} $|L(v)\cap\pal{v}|\ge 2\hatd_H(v)$, and
    \item\label[part]{part:list-size} $|L(v)\cap\pal{v}|\ge \hatd_H(v) + C\log^{1.1} n$ for some constant $C > 0$.
\end{enumerate} 
There exists an algorithm coloring every node of $H$ in $O(\log^* n)$ rounds of \bcongest with probability $1-n^{-\Theta(C)}$.
\end{restatable}

\Cref{lem:mct-bcongest} is a mere reformulation of \cite[Lemma 1]{HKNT22} with the notable exception that it works in \bcongest because of the additional \cref{part:list-known}. 
This allows the use of \emph{representative sets} \cite{HN23}.
At a high level, the technique is to save on the bandwidth necessary to send $\Theta(\log n)$ random colors by instead sending a pseudorandom sample. 
In \bcstream, it can be implemented with $O(\log^3 n)$ memory but requires more work. We refer interested readers to \cite[Section 7]{HN23}. The main idea is that a set of $\Theta(\log n)$ random colors can be represented by a random walk on an implicit expander graph.

\subsection{Concentration Inequalities}

We use the following variants of Chernoff bounds for dependent random variables. The first one is obtained, e.g., as a corollary of Lemma 1.8.7 and Theorems 1.10.1 and 1.10.5 in \cite{Doerr2020}.

\begin{lemma}[Martingales]\label{lem:chernoff}
Let $\{X_i\}_{i=1}^r$ be binary random variables, and $X=\sum_i X_i$.
Suppose that for all $i\in [r]$ and $(x_1,\ldots,x_{i-1})\in \{0,1\}^{i-1}$ with $\Pr\parens{X_1=x_1,\dots,X_r=x_{i-1}}>0$, $\Pr\parens{X_i=1\mid X_1=x_1,\dots,X_{i-1}=x_{i-1}}\le q_i\le 1$, then for any $\delta>0$,
\begin{equation}\label{eq:chernoffless}
\Pr\parens[\bigg]{X\ge(1+\delta)\sum_{i=1}^r q_i}
\le \exp\parens[\bigg]{-\frac{\min(\delta,\delta^2)}{3}\sum_{i=1}^r q_i}\ .
\end{equation}
Suppose instead that $\Pr\parens*{X_i=1\mid X_1=x_1,\dots,X_{i-1}=x_{i-1}}\ge q_i$, $q_i\in (0,1)$ holds for $i, x_1, \ldots, x_{i-1}$ over the same ranges, then for any $\delta\in [0,1]$,
\begin{equation}\label{eq:chernoffmore}
    \Pr\parens[\bigg]{X\le(1-\delta)\sum_{i=1}^r q_i}
    \le \exp\parens[\bigg]{-\frac{\delta^2}{2}\sum_{i=1}^r q_i}\ .
\end{equation}
\end{lemma}

\paragraph{Talagrand Concentration Bound.}
A function $f(x_1,\ldots,x_n)$ is  \emph{$c$-Lipschitz} iff changing any single $x_i$ affects the value of $f$ by at most $c$, and $f$ is  \emph{$r$-certifiable} iff whenever $f(x_1,\ldots,x_n) \geq s$ for some value $s$, there exist $r\cdot s$ inputs $x_{i_1},\ldots,x_{i_{r\cdot s}}$ such that knowing the values of these inputs certifies $f\geq s$ (i.e., $f\geq s$ whatever the values of $x_i$ for $i\not \in \{i_1,\ldots,i_{r\cdot s}\}$).

\begin{lemma}[Talagrand's inequality~\cite{Talagrand95,DP09}]
\label{lem:talagrand}
Let $\{X_i\}_{i=1}^n$ be $n$ independent random variables and $f(X_1,\ldots,X_n)$ be a $c$-Lipschitz $r$-certifiable function; then for $t\geq 1$,
\[\Pr\parens*{\abs*{f-\Exp[f]}>t+30c\sqrt{r\cdot\Exp[f]}}\leq 4 \cdot \exp\parens*{-\frac{t^2}{8c^2r\Exp[f]}}\]
\end{lemma}

\section{Algorithm and Analysis}
\label{sec:alg}

In this section, we describe our algorithm and give the main technical ideas behind \cref{thm:bcongest}. \cref{alg:high-level} gives a high-level description of our algorithm.

The main technical contribution is a $O(\log^* n)$-round algorithm for coloring graphs with $\Delta\in\Omega(\log^3 n)$. For low-degree graphs, a $O(\log^3\log n)$-round algorithm is known \cite{BEPSv3,GK21}. 
We conjecture that our algorithm actually shatters the graph in $O(\log^*n)$ rounds when $\Delta=O(\log^3 n)$. 
If this was to be true, \cite{BEPSv3} would no longer be required for small $\Delta$. 
This would make any improvement to the deterministic complexity of $(\deg+1)$-list-coloring, including beyond $o(\log n)$, carry over to our algorithm.

\begin{Algorithm}
\label{alg:high-level}
High Level Description of our Algorithm.

\noindent
\textbf{Parameters:} Let $C=O(1)$ be a large enough constant,
\begin{equation}
\label{eq:parameters}
\ell = C\log^{1.1} n\ , \quad \epsilon=10^{-5} \quad\text{and} \quad \beta=401  \ .
\end{equation} 

\begin{enumerate}[leftmargin=*]
    \item\label[step]{step:setting-up} \textbf{Setting up.} Compute an $\epsilon$-almost-clique decomposition $\Vsparse, K_1, \ldots, K_k$. 
    Compute outliers $O_K$ and inliers $I_K=K\setminus O_K$ in each clique $K$ (see \cref{def:outliers}), as well as put-aside sets $P_K$ (see \cref{lem:put-aside-build}). We define a value $x(K)=\Theta(\avganti_K+\avgext_K+\log n)$ for each clique (see \cref{eq:clique-palette}). By extension, let $x(v)=x(K)$ for each $v\in K$.

    Cliques are categorized as full, open, or closed (\cref{def:classification-cliques}).
    The following three steps aim at generating slack for each type:
    \begin{enumerate}[label=\textit{(\roman*)},leftmargin=0.5em]
        \item \emph{Slack Generation}: each node tries a color in $[\Delta+1]\setminus [x(v)]$ w.p. $\ps=1/200$.
        \item \emph{Colorful Matching}: by trying colors in $[\Delta+1]\setminus[x(K)]$ for $O(\beta)$ rounds, we color $\beta\avganti_K$ pairs of anti-edges in each $K$.
        \item \emph{Put-Aside Sets}: we find in each full clique sets $P_K\subseteq I_K$ of size $\Theta(\ell)$ such that $P_K$ has no edge to $P_{K'}$ for all $K\neq K'$.
    \end{enumerate}
    
    Each sparse node has $\Omega(\Delta)$ permanent slack from the slack generation step; hence, we color them in $O(\log^*n)$ rounds with \multitrial. We color outliers $O_K$ with colors from $[\Delta+1]\setminus [x(K)]$ with \multitrial using the $\Omega(\Delta)$ temporary slack provided by inactive inliers.
    
    \item\label[step]{step:sct} \textbf{Synchronized Color Trial.} In each clique, we compute the clique palette $\pal{K}$ and sample a permutation $\pi$ of $\hK\setminus P_K$. Each node $v\in \hK\setminus P_K$ tries the $\pi(v)$-th color of $\pal{K}$. In open cliques (see \cref{def:classification-cliques}), we run an extra $O(1)$ rounds of \trycolor using only colors from $[\Delta+1]\setminus [x(K)]$. 
    \item\label[step]{step:multitrial} \textbf{Completing the Coloring.} Uncolored nodes satisfy \[|[x(v)]\cap\pal{v}| \ge 2\hatd(v)\ .\] Put-aside sets ensure that every node has slack $\Omega(\ell)$; hence, inliers are colored in $O(\log^* n)$ rounds by \multitrial. 
    
    \item\label[step]{step:outliers} \textbf{Coloring Put-Aside Sets.}
    We color put-aside sets in two steps: first, we reduce their size to $O(\log n/\log\log n)$ by running \emph{non-adaptive} randomized color trial. Then, each node sends $|P_K|+1$ colors from a $\poly(\log n)$-sized set of colors. This takes $O(1)$ rounds: $O(\log n/\log\log n)\times O(\log\log n)$ bits to send.
\end{enumerate}
\end{Algorithm}

The key technical idea is to \emph{reserve} colors $\set{1,2,\ldots,x(K)}$ in each clique, where $x(K)$ is an integer that depends on the density of $K$ (see \cref{eq:clique-palette}). 
It is straightforward to see that reserve colors $[x(K)]$ are not used during \cref{step:setting-up,step:sct}. The value of $x(K)$ is chosen to be greater than nodes' degrees at the end of \cref{step:sct}. This allows using lists $L(v):=[x(v)]$ for the \multitrial in \cref{step:multitrial}.

\subsection{\texorpdfstring%
{\Cref{step:setting-up}: Setting up}%
{Step~\ref{step:setting-up}: Setting up}} 

Assume we have an $\epsilon$-almost-clique decomposition $\Vsparse$, $K_1$, $\ldots,$ $K_k$ (see \cref{def:almost-clique}). Sparse nodes can be colored in $O(\log^* n)$ rounds \cite{HN23}, so we focus our attention on almost-cliques. We call outliers the (possibly empty) set of nodes in each clique whose external degree or anti-degree derives more than a constant factor from the average.

\begin{definition}[Inliers/Outliers]
\label{def:outliers}
For each $K$, we define its set of \emph{outliers} as 
\begin{equation}
\label{eq:outliers}
    O_K=\set{v\in K: e_v \ge 30\avgext_K\text{ or } a_v \ge 30\avganti_K}\ .
\end{equation}
We call the remaining \emph{uncolored} nodes $I_K=\hK\setminus O_K$ \emph{inliers}.
\end{definition}

In each clique, outliers represent only a small fraction of the vertices; hence, can be colored beforehand with the temporary slack provided by their $\Omega(\Delta)$ uncolored neighbors in $I_K$. 

\begin{claim}
\label{claim:inliers}
For each $K$, after generating slack and computing a colorful matching, w.h.p.\ $|I_K| \ge 0.9\Delta$.
\end{claim}

\begin{proof}
By Markov inequality, outliers represent at most a $1/15$ fraction of $K$. Furthermore, nodes get colored during slack generation w.p. at most $\ps = 1/200$ (see \cref{lem:slackgeneration}). By Chernoff, w.h.p., at most a $1/100$ fraction of $K$ gets colored. The colorful matching comprises $2\beta\avganti_K \le 10^3\epsilon\Delta \le \Delta/100$ nodes by our choice of $\epsilon$. Therefore, $|I_K| \ge (1-1/15-1/100-1/100)|K| \ge 0.9\Delta$ by our choice of $\epsilon$ (\cref{eq:parameters}).
\end{proof}

We classify cliques in three categories, depending on the degree nodes have after \cref{step:sct}. Each type of clique receives slack from different sources: full cliques from put-aside sets, open cliques from the slack generation step, and closed cliques from the colorful matching.

\begin{definition}[Full/Open/Closed Cliques]
\label{def:classification-cliques}
For each $i\in[k]$, we say that $K=K_i$ is:
\begin{itemize}
    \item \emph{full} if $\avganti_K + \avgext_K < \ell$, where $\ell$ is defined in \cref{eq:parameters},
    \item \emph{open} if $K$ is not full and $2\avganti_K < \avgext_K$, and 
    \item \emph{closed} if $K$ is neither full nor open.
\end{itemize}
We denote by $\Kfull$ (respectively $\Kopen$ and $\Kclosed$) the set of full cliques (respectively open and closed cliques).
\end{definition}

In each clique, we reserve $x(K)$ colors depending on the clique's density. We will ensure that $[x(K)]\subseteq \pal{K}$ until we color inliers with \multitrial (\cref{step:multitrial}). For a clique $K$, define
\begin{equation} 
\label{eq:clique-palette}
x(K) = \begin{cases}
    200\ell & \text{if } K\in\Kfull\\
    400\avganti_K & \text{if } K\in\Kclosed\\
    \csp\epsilon/8\cdot\avgext_K &  \text{if } K\in\Kopen
\end{cases}\ ,
\end{equation}
where $\csp$ is the constant from \cref{lem:slackgeneration}.
By extension, we write $x(v) = x(K)$ for each $v\in K$.

\paragraph{Put-Aside Sets.}
Recall that to color in $O(\log^* n)$ rounds with \multitrial, nodes need slack at least $\ell = \Theta(\log^{1.1} n)$ (\cref{lem:mct-bcongest}, \cref{part:list-size}). Nodes from very dense cliques do not receive enough permanent slack from the slack generation phase. Following \cite[Section 5.4]{HKNT22}, we overcome this issue by putting aside sets of $\Theta(\ell)$ nodes in each very dense clique to provide temporary slack. These sets remain uncolored until the very end of the algorithm. These are necessary only in very dense cliques, whose nodes have $O(\ell)$ external neighbors. It allows us to find put-aside sets such that no edge connects sets from different cliques. The lack of connections allows us to color each set independently at the very end. 
See \cite[Lemma 5]{HKNT22} for a proof of \cref{lem:put-aside-build}.

\begin{restatable}[Put-Aside Sets]{lemma}{PutAsideLemma}
\label{lem:put-aside-build}
There exists a $O(1)$-round \bcongest algorithm finding subsets $P_K\subseteq I_K$ of size $201\ell$ in each almost-clique $K\in\Kfull$, such that $P_K$ has no edges to other $P_{K'}$ for $K'\neq K$.
\end{restatable}

\subsection{\texorpdfstring%
{\Cref{step:sct}: Synchronized Color Trial}%
{Step~\ref{step:sct}: Synchronized Color Trial}}

The idea of the following \cref{lem:sct} (which is a reformulation of \cite{HKNT22}) is to distribute a set of colors to nodes in the clique. Each color has a unique recipient. This avoids in-clique conflicts, and a node can only fail to adopt the color it received due to its external neighbors. Therefore, the expected number of nodes to fail is $\sum_{v\in K} O(e_v/\Delta) = O(\avgext_K)$. 

\begin{restatable}[{\cite[Section 5.5]{HKNT22}}]{lemma}{sctLemma}
\label{lem:sct}
Let $x$ be an integer, $K$ be a clique, and $S=\hK\setminus P_K$ be such that $0.75\Delta \le |S| \le |\pal{K}| - x$. Suppose $\pi$ is a uniform permutation of $[|S|]$. If for each $i\in[|S|]$ the $i$-th node in $S$ tries the $\pi(i)$-th color in the set $\pal{K}\setminus[x]$, then w.h.p.\ the number of nodes to remain uncolored is $8\max\set{6\avgext_K, C\log n}$. This holds even if the random bits outside of $K$ are chosen adversarially.
\end{restatable}

\cref{lem:sct-enough} shows that each clique has enough colors, even if when we reserve $x(K)$ colors.

\begin{lemma}
\label{lem:sct-enough}
For all $K$, $|\pal{K}|-x(K) \ge |\hK\setminus P_K|$.
\end{lemma}

\begin{proof}
We consider each type of clique separately.
In a full clique $K$, recall that we computed a set of put-aside nodes $P_K$ of size $201\ell=\Theta(\log^{1.1} n)$ that must remain uncolored (\cref{lem:put-aside-build}). The set $S$ of nodes to try a color during the synchronized color trial is $|S|=|\hK\setminus P_K| \ge 0.75\Delta$ (by of \cref{claim:inliers} and $\Delta \gg \ell$). 
The number of colors used in $K$ is bounded by the number of colored nodes; hence, $|\pal{K}| \ge \Delta - (|K|-|\hK|)$. 
Since each full clique has size at most $\Delta + \ell$, we infer $|\pal{K}| \ge |\hK|-\ell$. Put-aside sets have size $|P_K|=201\ell$, so 
\[ |\hK\setminus P_K| = |\hK| - 201\ell
    \le |\pal{K}| - 200\ell 
    = |\pal{K}|-x(v)\ .  \tag{by \cref{eq:clique-palette}} \]

Suppose that $K$ is open, i.e.  $\avganti_K \le \avgext_K/2$ (\cref{def:classification-cliques}). 
By summing on each $v\in K$ over the bounds $\Delta \ge |K\cap N(v)| + e_v$ and $|K| = |K\cap N(v)| + a_v$, we get $\Delta - |K| \ge \avgext_K - \avganti_K \ge \avgext_K/2$. By our choice of $x(K)$, \[ |\pal{K}| - x(K) \ge |\hK| + \avgext_K/2 - x(K) \ge |\hK|\ . \]

Suppose now that $K$ is closed. Denote by $t$ the number of nodes colored during the slack generation step or as outliers. In closed clique, we compute a colorful matching of size $\beta\avganti_K$. Hence $|\pal{K}| \ge \Delta - t - \beta\avganti_K$. On the other hand, each edge in the matching colors two nodes. Therefore, the number of uncolored nodes is 
\begin{align*}
    |\hK| &\le |K| - t - 2\beta\avganti_K \\
          &\le (\Delta - t - \beta\avganti_K) - (\beta-1)\avganti_K \tag{because $|K|\le \Delta+\avganti_K$} \\
          &\le |\pal{K}| - x(K)\ . \tag{by definition of $\beta$, \cref{eq:parameters}} \qedhere
\end{align*}
\end{proof}

We now claim that each node has enough slack after SCT. 
Details of its implementation and related proofs are postponed to a later section (\cref{sec:implem-sct}, \cref{lem:learn-clique-palette,lem:sample-perm-constant}).

\begin{lemma}
\label{lem:deg-sct}
At the end of \cref{step:sct}, w.h.p.\ each $v\in \hK$ satisfies $|[x(v)]\cap\pal{v}| \ge 2\hatd(v)$.
\end{lemma}

\begin{proof}\renewcommand{\qed}{}
By \cref{lem:sct-enough}, cliques carry more colors than nodes they try to color during SCT, and by \cref{lem:sct}, at most $O(\avgext_K + \log n)$ nodes remain uncolored per clique. Simple counting shows the following claim.

\begin{claim}
\label{claim:deg-sct}
After the synchronized color trial, every uncolored $v\in K$ satisfies
\begin{itemize}
    \item $2\hatd(v) + e_v \le x(v)$ if $v\in\Kfull\cup \Kclosed$, and
    \item $\hatd(v) \le 80\avgext_K$ if $K\in\Kopen$.
\end{itemize}
\end{claim}

\begin{proof}
Let $v\in K$ and assume first $K\in\Kfull$. Since only inliers remain to be colored, $e_v \le 30\avgext_K \le 30\ell$ (by \cref{eq:outliers}) and after the synchronized color trial at most $48\ell$ nodes remain uncolored in $K$ (by \cref{lem:sct}). Overall, $\hatd(v) \le 80\ell$ and $2\hatd(v) + e_v \le 200\ell = x(v)$ (by \cref{eq:clique-palette}). 
If $K\notin\Kfull$, by a similar argument $\hatd(v) \le e_v + 50\avgext_K \le 80\avgext_K$.
If $K\in\Kclosed$, then $\hatd(v) \le 80\avgext_K \le 160\avganti_K$ because $\avgext_K \le 2\avganti_K$. Hence, $2\hatd(v) + e_v \le 400\avganti_K = x(v)$.\Qed{claim:deg-sct}
\end{proof}

Observe that, since $x(v)$ has the same value for each $v\in K$, and colors from $[x(K)]$ are not used to color nodes of $K$, the only reason some $c\in[x(v)]$ might not belong to $\pal{v}$ is if it is used by an external neighbor of $v$. For all $v\in K$ with $K\in\Kfull\cup\Kclosed$, \cref{eq:slack-lists} follows from \cref{claim:deg-sct}:

\begin{equation}
\label{eq:slack-lists}
|[x(v)]\cap\pal{v}| \ge x(v) - e_v \ge 2\hatd(v)\ . 
\end{equation}

For $v\in K$ with $K\in\Kopen$, we need $O(1)$ additional rounds of \trycolor to ensure \cref{eq:slack-lists}. However, we need to preserve $[x(K)]\subseteq \pal{K}$. Thus, nodes of $K$ try random colors in $\pal{v}\setminus [x(v)]$. We now show it is enough to reduce the uncolored degree.

Let $v\in K$ for any $K\in\Kopen$. By \cref{claim:deg-sct}, $\hatd(v) \le 80\avgext_K$; we show that $|\pal{v}| - x(v) \ge \Omega(\avgext_K)$. By \cref{claim:rct-less-colors}, even when using only colors from $\pal{v}\setminus[x(v)]$, after one call to \trycolor the uncolored degree of each node decreases by a constant factor. After $O(1)$ rounds, with high probability, the uncolored degree of each $v$ verifies the desired equation.

\begin{claim}
\label{claim:small-ext}
For each $v\in K$, $\Delta-d(v)+e_v \ge \avgext_K/2$.
\end{claim}

\begin{proof}
Since $\Delta \ge |K\cap N(v)| + e_v$ and $|K| = |K\cap N(v)| + a_v$, we have $\Delta \ge |K| + e_v - a_v$.
We must have $|K| \le \Delta - \avgext_K/2$ for, otherwise, summing on all $v\in K$, we get $\avganti_K \ge |K|-\Delta + \avgext_K > \avgext_K/2$. Now, for $v\in K$, we have $|N(v)\cap K| \le |K| \le \Delta - \avgext_K/2$. The claim follows.\Qed{claim:small-ext}
\end{proof}

If $e_v \le C\log n$, by \cref{claim:small-ext}, $s(v) \ge \Delta-d(v) \ge \avgext_K/2 - C\log n \ge \avgext_K/3$ because $\avgext_K \ge \ell/2 \gg C\log n$. If $e_v \ge C\log n$, vertex $v$ receives $\gamma\epsilon/2 \cdot e_v$ permanent slack from the slack generation step w.p. $1-n^{-\Theta(C)}$ (by \cref{lem:slackgeneration}). Overall, nodes use lists of size

\begin{align*}
    |\pal{v}| - x(v) &\ge \Delta - d(v) + \csp\epsilon/2\cdot e_v - x(v) \\
    &\ge \csp\epsilon/2 \cdot (\Delta - d(v) + e_v) - x(v) \tag{$\csp\epsilon/2 < 1$} \\
    &\ge \csp\epsilon/4\cdot \avgext_K - x(v) \tag{by \cref{claim:small-ext}} \\
    &\ge \csp\epsilon/8\cdot\avgext_K\ . \tag{by \cref{eq:clique-palette}}
\end{align*}

By \cref{claim:rct-less-colors} with $\alpha = \csp\epsilon/640$, after \trycolor the uncolored degree of each node reduces by a constant factor with high probability.
\Qed{lem:deg-sct}
\end{proof}

\subsection{\texorpdfstring%
{\Cref{step:outliers}: Coloring Put-Aside Sets}%
{Step~\ref{step:outliers}: Coloring Put-Aside Sets}}
\label{sec:put-aside}

Our goal, in this section, is to reduce the size of put-aside sets to $O(\log n/\log\log n)$. Once this is achieved, coloring their remaining nodes only takes $O(1)$ rounds, as the next lemma shows.

\begin{lemma}
\label{lem:color-put-aside}
Suppose all nodes are colored except put-aside sets $P_K$ in each $K\in\Kfull$ of size $O(\log n/\log\log n)$. Then, w.h.p.\ we can complete the coloring in $O(1)$ rounds of \bcongest.
\end{lemma}

\begin{proof}
Recall that no edges exist between put-aside sets. Hence, we color each put-aside set independently. 
We can assume without loss of generality that $|\pal{K}|=O(\log^3 n)$. Indeed, since nodes have $O(\log^{1.1}n)$ external and anti-degree, 
any $D\subseteq\pal{K}$ of size $\Theta(\log^3 n)$ works as replacement for the clique palette when $\pal{K}$ is larger. Nodes use \cref{alg:learnpalette} to learn $\pal{K}$ in $O(1)$ rounds (\cref{lem:learn-clique-palette}).

Therefore, describing a color $c\in \pal{K}$ takes $O(\log\log n)$ bits. If $\avganti_K \ge C\log n$, the clique palette has enough colors for every node, i.e., $|\pal{K}\cap\pal{v}| \ge |P_K|+1$. If $\avganti_K < C\log n$, lists $L(v)=\pal{K}\cup C(K\setminus N(v))$ have $|P_K|+1$ colors (\cref{claim:clique-palette} with an empty matching and $a_v$ extra colors). Since lists have size $|P_K|+1=O(\log n/\log \log n)$ and each color takes $O(\log\log n)$ bits, nodes can broadcast their list in $O(1)$ rounds. Nodes complete the coloring without additional communication, simulating a greedy sequential algorithm with the lists.
\end{proof}

The following technical claim (which is a direct application of Chernoff) allows us to assume we have global communication within almost-clique if the number of messages to send is small enough.
In particular, nodes can learn all the identifiers from $P_K$, therefore relabel nodes with $O(\log\log n)$-bit.

\begin{claim}[Many-to-All Broadcast]
\label{lem:echo}
Let $K$ be an almost-clique with $O(\Delta/\log n)$ nodes with an $O(\log n)$-bit message to send to everyone in $K$. Suppose each node with a message broadcasts it, before each node in $K$ broadcasts $O(1)$ messages it received, picked randomly. Then, w.h.p., all messages are received by every node in $K$.
\end{claim}

\newcommand{\comptry}{\alg{CompressTry}}

The key difficulty in coloring put-aside sets lies in reducing their sizes to $O(\log n / \log \log n)$.
We use a procedure \comptry, which simulates a sequential algorithm where nodes of the put-aside set, in the order of their {\ID}s, each perform $O(\log n/\log\log n)$ times a \emph{non-adaptive} \trycolor with slack $z$. The following technical lemma analyzes the performance of \comptry.
We defer the exact description of \comptry and proof of \cref{lem:reduce-put-aside} to \cref{sec:reduce-put-aside}.

\begin{restatable}{lemma}{reducePutAside}
\label{lem:reduce-put-aside}
Let $K\in\Kfull$ and fix a set $S\subseteq \hK$ of size $O(\log^{1.1} n)$. 
Furthermore, suppose each $v\in S$ has a list $L(v)$ of at most 
$C\log^{1.1} n$ colors known to every $u\in S$, and such that $|L(v)\cap \pal{v}| \ge |S|+z$ for a fixed $z \ge C\log n/\log\log n$. 
Then, w.p.\ $1-e^{-z}-1/\poly(n)$, \comptry colors all but $z$ nodes in $S$. Furthermore, \comptry uses $O(\log n/\log\log n)$ bandwidth.
\end{restatable}

\cref{lem:small-put-aside} shows how we use \comptry to reduce the size of the put-aside sets. In cliques with colorful matching, nodes have $\avganti_K \in \Omega(\log n)$ slack; \comptry directly reduces $P_K$ to $O(\log n/\log\log n)$ nodes by using the clique palette. In cliques where $\avganti_K < C\log n$, we first put-aside $O(\log n)$ nodes to reduce $P_K$ to $O(\log n)$ using the clique palette. Then, nodes add colors used by their anti-neighbors to their list, and \comptry finishes to reduce $P_K$ to $O(\log n/\log\log n)$.

\begin{lemma}
\label{lem:small-put-aside}
There is a $O(1)$-round \bcongest algorithm reducing the number of uncolored nodes in $P_K$ to $O(\log n/\log\log n)$ with high probability.
\end{lemma}

\begin{proof}
For cliques such that $\avganti_K \ge C\log n$, \cref{lem:reduce-put-aside} allows us to directly reduce $P_K$ to a set of size $z:=C\log n/\log\log n$. This is because, in such cliques, we compute a colorful matching of size $\beta\avganti_K \ge \avganti_K + a_v$, for each $v \in P_K$ (which are inliers). Therefore, using lists $L(v):=\pal{K}$, by \cref{claim:clique-palette}, $|L(v)\cap \pal{v}| \ge |P_K|+\avganti_K \ge |P_K| + z$. Note that the clique palette can be publicly learned in $O(1)$ rounds by \cref{lem:learn-clique-palette}. \comptry succeeds only w.p.\ $1-e^{-z}$, but by repeating independently $\log\log n$ times, the probability that at least one instance succeeds is $1-e^{-z\log\log n} = 1-n^{-C}$.
Overall, we need $\log\log n\times O(\log n/\log\log n)=O(\log n)$ bandwidth.

Henceforth, we assume that $\avganti_K < C\log n$. The main difference is that we do not have a colorful matching, so the clique palette does not approximate $\pal{v}$ well. We settle this in two steps.

\emph{From $O(\log^{1.1} n)$ to $O(\log n)$.}
Let $S\subseteq P_K$ be an arbitrary subset of $P_K$ of $31C\log n$ nodes. By \cref{claim:clique-palette}, $|\pal{K}\cap\pal{v}|\ge |P_K|-a_v \ge |P_K\setminus S| + C\log n$. Therefore, \comptry with lists $L(v)=\pal{K}$ and $z=C\log n$ reduces $P_K$ w.h.p.\ to size $32C\log n$ (the $C\log n$ nodes left uncolored in $P_K\setminus S$ by \comptry and the $31C\log n$ uncolored nodes of $S$).

\emph{From $O(\log n)$ to $O(\log n/\log\log n)$.}
Now, instead of using only the clique palette, we augment lists with colors of anti-neighbors. Let $L(v):=\pal{K}\cup \col(K\setminus N(v))$. Since we are adding $a_v$ colors to each list, \cref{claim:clique-palette}, even with an empty matching, gives us, $|L(v)\cap\pal{v}| = |\pal{K}\cap\pal{v}| + a_v \ge |P_K|$. If we now put-aside a set $S\subseteq P_K$ of $z:=C\log n/\log\log n$ nodes, lists $L(v)$ verify $|L(v)\cap\pal{v}|\ge |P_K\setminus S| + z$. To conclude, it remains to explain how nodes learn lists $L(v)$.

Since $\avganti_K < C\log n$, each node has at most $30C\log n$ anti-neighbors in the clique. If we relabel nodes of $P_K$ using identifiers in $[|P_K|]$ (with \cref{lem:echo}), every $u\in K$ can describe the set $P_K\setminus N(v)$ with a bit-map in one $O(\log n)$-bit message. Note that only $O(\log^2 n)$ nodes will need to send a bit-map, i.e.\ at most $O(\log n)$ per node in $P_K$. By \cref{lem:echo}, all messages can be disseminated in $O(1)$ rounds to all nodes in $K$. Thus, all lists are known and we make $\log\log n$ independent calls to \comptry.
\end{proof}

\subsection{\texorpdfstring%
{Proof of \cref{thm:bcongest}}%
{Proof of Theorem~\ref{thm:bcongest}}}
\label{sec:proof}

By \cref{lem:acd-bcongest}, we can compute the almost-clique decomposition in $O(1)$ rounds. By aggregation on a depth-2 BFS tree, nodes in each clique can count $\avganti_K$ and $\avgext_K$, thus know to which category their clique belongs to, as well a their value of $x(K)$. Then, with w.p.\ $\ps$ every node decides independently to try a color in $[\Delta+1]\setminus[x(v)]$ (for consistency, let $x(v)=0$ for all $v\in\Vsparse$). Finally, in each clique with $\avganti_K \ge C\log n$, we compute a colorful matching of size $\beta\avganti_K$ (by \cref{thm:colorful-matching}). By \cref{lem:put-aside-build}, we compute put-aside sets $P_K$ in $O(1)$ rounds.

\paragraph{Sparse Nodes \& Outliers.} Each $v\in\Vsparse$ has permanent slack $\Omega(\Delta)$ (by \cref{lem:slackgeneration} and because they are $\Omega(\Delta)$-sparse). Hence, we color $\Vsparse$ in $O(\log^* n)$ rounds of \multitrial (by \cref{lem:mct-bcongest}). Since nodes know $\avganti_K$ and $\avgext_K$, they can tell if they are outliers. Outliers have slack $(0.9-\epsilon)\Delta \ge \Delta/2$ from inactive inliers neighbors (\cref{claim:inliers}). Contrary to sparse nodes, we must avoid coloring outliers of $K$ with colors from $[x(K)]$. By definition $x(K) = 10^3\epsilon\Delta$ (\cref{eq:clique-palette}); by our choice of $\epsilon$, outliers have slack $(1/2-10^3\epsilon)\Delta \ge \Delta/3$ even when trying colors from $[\Delta+1]\setminus[x(K)]$. By \cref{claim:rct-less-colors,lem:mct-bcongest}, outliers are colored in $O(\log^* n)$ rounds with high probability.

\paragraph{Inliers.} 
Henceforth, we condition on the success of \cref{step:setting-up,step:sct} in every clique. By \cref{lem:deg-sct}, each inlier satisfies $|L(v)\cap\pal{v}| \ge 2\hatd(v)$ with $L(v):=[x(v)]$. To run \multitrial, we need lists to intersect the palette on at least $\Omega(\ell)=\Omega(C\log^{1.1} n)$ colors (\cref{lem:mct-bcongest}, \cref{part:list-size}). If $v$ is in a open or closed clique, then $\avganti_K$ or $\avgext_K$ is greater than $\ell/2$ and $|L(v)\cap\pal{v}| \ge x(v) - e_v \ge \hatd(v) + \Omega(\ell)$ (by \cref{eq:clique-palette}). On the other hand, if $v$ is in a full clique, then $a_v \le 30\avganti_K \le 30\ell$ (by \cref{eq:outliers} and \cref{lem:put-aside-build}). Therefore, $v$ has at least $|N(v)\cap P_K|\ge \ell$ temporary slack from inactive put-aside neighbors (by \cref{lem:put-aside-build}).
Finally, it suffices to broadcast $x(v)$ for all neighbors of $v$ to learn $L(v)$. Therefore, lists $L(v):=[x(v)]$ verify all properties requires to run \multitrial in \bcongest (\cref{lem:mct-bcongest}). With high probability, all nodes are colored in $O(\log^* n)$ rounds -- except put-aside sets. By \cref{lem:color-put-aside,lem:small-put-aside}, we can color put-aside sets in $O(1)$ rounds. 
\Qed{thm:bcongest}

\section{Synchronized Color Trial in \bcongest}
\label{sec:implem-sct}

At its core, synchronized color trial is simply about creating a random bijection between (most of) a set of colors and (most of) the uncolored nodes of a clique. Our implementation uses the clique palette as a set of colors and randomly permutes the nodes. The order of each node in the permutation tells it which color to take in the clique palette. This entails two difficulties. Firstly, to make use of its order in the sampled permutation, each node needs to know the matching color in the clique palette. We show that $O(1)$ rounds of \bcongest suffice for all nodes to learn their clique palette. The second issue is sampling the permutation, and entails a more involved process. For simplicity, we describe first a $O(\log \log n)$-round permutation sampling procedure, which suffices for \cref{thm:bcongest,thm:streaming}. We then explain how to reduce it down to $O(1)$ rounds with a slightly more involved procedure.

We will need the following technical lemma.

\begin{lemma}
\label{lem:dom-groups}
Let $K$ be an almost-clique and an integer $k\le \Delta/(C\log n)$ for some large enough $C > 0$. Suppose each $v\in K$ samples $t(v)\in[k]$ uniformly at random. Then, with high probability, for each $i\in[k]$, the set $T_i=\set{v\in K: t(v)=i}$ satisfies that for any $u,w\in K$, $|T_i\cap N(u)\cap N(w)| \ge (C/4)\log n$. We say that $T_i$ \emph{2-hop connects} $K$ in that each pair of nodes in $K$ has a common neighbor in $T_i$.
\end{lemma}

Note that since $T_i \subseteq K$, each $T_i$ also 2-hop connects itself, thus has diameter $2$.

\begin{proof}
Fix an index $i\in[k]$. Each node joins $T_i$ w.p.\ $1/k$ independently from other nodes. For each pair $u,w \in K$, in expectation, $T_i\cap N(u) \cap N(w)$ has size $\mu=\card{N(u)\cap N(v)}/k \ge (1-2\eps)\Delta/k \ge (C/2)\log n$. By a classic Chernoff bound,
$\Pr\parens{\card{T_i\cap N(u) \cap N(w)} \leq \mu/2} \le \exp(-\mu/12) \le 1/\poly(n)$. By union bound, w.h.p., we have $\card{T_i\cap N(u) \cap N(w)} \geq \Delta/(4k)$ for all $i\in[k]$ and $u,w \in K$.
\end{proof}

\paragraph{Learning the clique palette.} We learn the clique palette by dividing the color space into $O(\Delta/\log n)$ contiguous subpalettes. Given a 2-hop connecting set of nodes to handle each subpalette -- with a trivial construction due to \cref{lem:dom-groups} -- each node learns $\pal{K}$ in $O(1)$ rounds.
Recall that $\col(S)$ denotes the set of colors currently assigned to a set $S$ of nodes.

\begin{Algorithm}
\label{alg:learnpalette}
Procedure \learn, in almost-clique $K$.

\noindent
\textbf{Parameters:} Let $C=O(1)$ be a large enough constant, $k = \floor{\Delta/(C\log n)}$.
\smallskip

\noindent
Assume $K$ to be split into $k$ 2-hop connecting sets $T_1,\ldots, T_k$. Let $R_i := \set{1+ \floor{(i-1) \cdot (\Delta+1) / k}, \floor{i \cdot (\Delta+1) / k}}$, i.e., $R_1,\ldots,R_k$ partition the color space $[\Delta+1]$.

\begin{enumerate}[leftmargin=*]
    \item\label[step]{step:learnin} Each $v$ encodes $R_{t(v)} \cap \col(N(v)\cap K)$ into a $C\log n$-sized bit-map and broadcasts it.
    \item\label[step]{step:teachout} 
    For each $i\in [k]$, each $v\in K$ combines the bit-maps received from its neighbors in $T_i$, i.e., computes
    \[\bigcup_{u \in N(v) \cap T_i} \parens[\bigg]{R_{i} \cap \col(N(u)\cap K)}\] and takes it for $R_{i} \cap \col(K)$.
\end{enumerate}
\end{Algorithm}
\begin{lemma}
\label{lem:learn-clique-palette}
Let $K$ be an almost-clique of palette $\pal{K}$. \learn has each $v\in K$ learn $\pal{K}$ in $O(1)$ rounds of \bcongest.
\end{lemma}

\begin{proof}
In $\Delta+1$-coloring, learning $\pal{K}$ is equivalent to learning the \emph{used} colors $\col(K)$.
\learn requires $O(1)$ rounds of \bcongest, as each node in $K$ only sends one $C \log n$-bit message. Let us consider a color $c \in\col(K)$, a node $v\in K$, and argue that $v$ learn $c$. Let $R_i$ be such that $c \in R_i$, and $u\in K$ a node with color $c$. Since $T_i$ 2-hop connects $K$, there exists a node in $T_i \cap N(u) \cap N(v)$. Such a node contains $c$ in the bitmap it computes in \cref{step:learnin} of \learn, and $v$ receives this bitmap in \cref{step:teachout}. As this works for every $c \in \col(K)$ and $v\in K$, all $v\in K$ learn $\col(K)$.
\end{proof}

\paragraph{Sampling the permutation.} At a high level, the $O(\log \log n)$ algorithm for permuting the nodes presented in this section has the nodes undergo two shuffling steps. Nodes first undergo a ``rough shuffling'', which puts them into buckets, roughly positioning them in the permutation. Each group then does a ``fine shuffling'' to give each node its exact position. 

An important step in both our $O(\log \log n)$ and our $O(1)$ implementation is giving nodes $O(\log \log n)$-bit labels unique within their buckets. Using the smaller labels instead of the  original node {\ID}s allows each bucket to save a multiplicative $\Theta(\log n / \log \log n)$ factor when describing a permutation of its elements.

\begin{Algorithm}
\label{alg:relabel} Procedure \relabel, in 2-hop connected set of nodes $T \subseteq V$, for subset $S \subseteq T$.

\noindent
\textbf{Parameters:} Let $C=O(1)$ be a large enough constant, $x := \ceil{C \log n / \log \log n}$.

\begin{enumerate}[leftmargin=*]    
   \item\label[step]{step:samp-labels} Each $v\in S$ samples and broadcasts $x$ labels in $[\card{S}^2 \log n]$, picked u.a.r.\ and independently.

    \item\label[step]{step:check-labels} Each $v\in T$ broadcasts an $x$-sized bit-map indicating, for each $j\in [x]$, whether multiple nodes in $S\cap N(v)$ have the same $j$th label. 

    \item\label[step]{step:take-labels} If for a minimum $j\in [x]$, all nodes in $S$ have distinct $j$th labels, $S$ uses them as new labels.
\end{enumerate}
\end{Algorithm}

\begin{lemma}
    \label{lem:relabel}
    Suppose $S$ has size $\poly(\log n)$. \relabel succeeds at relabeling $S$ in $O(1)$ \bcongest rounds, w.h.p.
\end{lemma}
\begin{proof}
    First, note that $O(1)$ \bcongest rounds suffice to compute $\card{S}$ for \cref{step:samp-labels}, as $T$ is $2$-hop connected.
    Since $\card{S}^2\log n \in \poly(\log n)$, each label sent by a node $v\in S$ during \cref{step:samp-labels} is representable with $O(\log \log n)$ bits. Thus, $x \in O(\log n / \log \log n)$ labels can be transmitted in $O(1)$ rounds.

    As $T$ 2-hop connects itself (a fortiori $S$), two nodes of $S$ with a common $j$th label are necessarily detected by a common neighbor during \cref{step:check-labels}. Taking the AND of all $x$-sized bitmaps sent in this step, the nodes in $T$ all learn for which $j \in[x]$ it holds that all nodes of $S$ picked distinct $j$th labels.

    We now analyze the probability that the relabeling succeeds, i.e., that a $j\in [x]$ as used in \cref{step:take-labels} exists. For each $j\in [x]$, each $j$th sampled label in $S$ has probability less than $1/(\card{S}\log n)$ of conflicting with one of the other $\card{S}-1$ $j$th labels. Hence, by union bound, the $j$th labels have a collision with probability at most $1/(\log n)$. Having $x$ independent samples implies success with probability at least $1-(\log n)^{-x} = 1-2^{-x \log \log n} = 1-2^{-C \log n} = 1-n^{-C}$, i.e., w.h.p.
\end{proof}

\begin{Algorithm}
\label{alg:permute-loglog}
Procedure \permute, in almost-clique $K$, on subset $S \subseteq K$ of the nodes.

\noindent
\textbf{Parameters:} Let $C=O(1)$ be a large enough constant, 
\[ k := \floor{\Delta/(C\log n)},\quad\text{and}\quad x := \ceil{C \log n / \log \log n} \ .\]

\begin{enumerate}[leftmargin=*]
    \item\label[step]{step:rough-permute} \textbf{Rough bucketing.} Each $v \in K$ independently picks a random $t(v) \in [k]$ u.a.r.

    For each $i\in [k]$, let $T_i:=\set{v \in K : t(v) = i}$ and $S_i := T_i \cap S$.
    
    \item\label[step]{step:prefix} \textbf{Counting buckets.} For each $i\in [k]$, the nodes in $T_i$ compute and broadcast $\card{S_i}$.

    \item\label[step]{step:relabel} \textbf{Relabeling.}
    Within each $T_i$, $i\in [k]$, use \relabel on $S_i$.

    \item\label[step]{step:fine-permute} \textbf{Permuting within buckets.} Within each $T_i$, the maximum \ID node gathers the new labels of $S_i$, picks a random permutation $\rho_i$ of $S_i$, and sends it to $T_i$, all along a BFS tree.

    \item\label[step]{step:permute-output} \textbf{Output.} Each $v \in S_i$ takes $\pi(v) := \rho_i(v) + \sum_{j < i} \card{S_j}$ as its index in the output $\pi$.
\end{enumerate}
\end{Algorithm}

\begin{lemma}
\label{lem:sample-perm-loglog}
With high probability, \permute outputs a permutation of $S$ in $O(\log\log n)$ rounds. For each permutation $\pi$ of $S$, the probability of sampling $\pi$ is bounded by $\frac{1}{ (1-1/\poly(n)) \cdot \card{S}!}$
\end{lemma}

\begin{proof}
By \cref{lem:dom-groups}, the sets $T_i$ computed in \cref{step:rough-permute} 2-hop connect $K$, w.h.p., and in particular have diameter $2$. Assuming this holds, \cref{step:prefix} only takes $O(1)$ rounds using a  aggregation and dissemination on the depth-2 BFS tree within each $T_i$. This allows each $v\in S_i$ to compute $\sum_{j<i} \card{S_j}$ for the last step of the algorithm.

In addition, it also holds w.h.p.\ that each $S_i \subseteq T_i$ has size $O(\log n)$. Assuming this holds, running \relabel in \cref{step:relabel} only requires $O(1)$ rounds per \cref{lem:relabel}, and it succeeds w.h.p.
Finally, the process takes $O(\log \log n)$ rounds due to \cref{step:fine-permute}, during which a leader node within each $T_i$ broadcasts $O(\log n)$ labels of $O(\log \log n)$ bits each.

We now argue the approximate uniformity of the sampling. Consider the random process in which each node in $S$ picks a random ordered bucket independently and u.a.r, and then each bucket is permuted uniformly at random. Let $\mu$ be the distribution of the permutation generated by this process. Clearly, $\mu$ is the uniform distribution. \permute is the same as this process, except it does not output anything if some high probability event $\cE$ does not hold. More precisely, the high probability event $\cE$ corresponds to all buckets being $2$-connected, all buckets being of $O(\log n)$ size, and \relabel succeeding.
Let $\mu_1$ be the distribution $\mu$ conditioned on $\cE$ holding, and $\mu_2$ be $\mu$ conditioned on $\cE$ not holding. Distribution $\mu_1$ is the output distribution of \permute, and we have $\mu = (1-1/\poly(n)) \mu_1 + (1/\poly(n))\mu_2$. Thus, for each permutation $\pi$, $\mu_1(\pi) \leq \mu(\pi)/(1-1/\poly(n)) = 1/((1-1/\poly(n))\card{S}!)$.%
\end{proof}

\paragraph{Reducing the complexity to a constant.} 
Our $O(1)$ implementation improves on the running time by splitting buckets from the first ``rough shuffling'' into sub-buckets, and arguing that most such buckets satisfy properties allowing them to use a leader to permute themselves as in \cref{alg:permute-loglog}, while buckets that fail this second sub-bucketing are few enough that they can be efficiently permuted with the help of the whole almost-clique. 

\begin{restatable}{lemma}{ConstantPermuteLemma}
\label{lem:sample-perm-constant}
There is an algorithm simulating the permutation sampling step of the synchronized color trial in $O(1)$ rounds of \bcongest.
\end{restatable}

 A key ingredient in the improved version of our algorithm is strengthening the properties satisfied by buckets. We will aim for the random subsets of almost-cliques formed to themselves have the properties almost-cliques. Let a \emph{$k$-bucketing} $t$ of a set of nodes $K$ be an assignment of a value $t(v) \in [k]$ to each node $v \in S$, defining sets $T_i := \set{v \in K : t(v) = i}$ for each $i \in [k]$. In our improved $O(1)$ algorithm, we also perform a second $k'$-bucketing $t'$ of each set $T_i$, defining sets $T_{i,i'} := \set{v \in T_i: t'(v) = i'}$ for each $(i,i') \in [k] \times [k']$.

\begin{definition}[Almost-clique-like, almost-clique-preserved] For a set of nodes $K$,

\begin{itemize}
    \item For $\eps \in (0,1/2)$, $K$ is said to be \emph{$\eps$-almost-clique-like} ($\eps$-AC-like) if $\forall v \in K, \card{N(v) \cap K} \geq (1-\eps)\card{K}$.

    \item For integers $i$, $k$ with $i\leq k$, a $k$-bucketing of $K$ is said to \emph{$\eps'$-almost-clique-preserve} ($\eps'$-AC-preserve) its $i$th bucket $T_i$ iff $\forall i\in [k], \forall v \in K$, $\card{N(v) \cap T_i} \in (1 \pm \eps') \card{N(v)\cap K}/k$. The bucketing is said to be \emph{$\eps'$-almost-clique-preserving} ($\eps'$-AC-preserving) if it $\eps'$-AC-preserves its $k$ buckets.
\end{itemize}
\end{definition}

\begin{lemma}
    \label{lem:preserving-ac-like}
    Let $\eps,\eps'$ be two positive constants s.t.\ $\eps+\eps'<1/2$. Let $K$ be $\eps$-AC-like, and let $T$ be an $\eps'$-AC-preserved bucket of a $k$-bucketing of $K$. Then, $T$ has size $\card{T} \in (1\pm2(\eps+\eps'))\card{K}/k$ and is $2(\eps+\eps')$-AC-like.
\end{lemma}
\begin{proof}
    For each $v \in K$, the bounds on $\card{K}$, $\card{N(v)\cap T}$, and $\card{N(v)\cap K}$ from $K$ being $\eps$-AC-like and $T$ being $\eps'$-AC-preserved yield:
    
    \[
    (1-\eps')(1-\eps)\frac{\card{K}}{k} 
    \leq (1-\eps')\frac{\card{N(v)\cap K}}{k}
    \leq \card*{N(v) \cap T}
    \leq (1+\eps')\frac{\card{N(v)\cap K}}{k}
    \leq (1+\eps')\frac{\card{K}}{k}
    \ .\]
    
    Counting edges between $T$ and $K$ two ways gives:
    
    \[
    \sum_{v \in K} \card{N(v) \cap T} 
    = \sum_{v \in T} \card{N(v) \cap K}
    \ .\]
    
    The combination of $(1-\eps)\card{T}\cdot \card{K} \leq \sum_{v \in T} \card{N(v) \cap K} \leq \card{T}\cdot \card{K}$ (from $K$ being $\eps$-AC-like) with the previous bounds on $\card{N(v)\cap T}$ gives as bounds on $\card{T}$:
    
    \[
    (1-\eps-\eps') \frac{\card{K}}{k}
    \leq (1-\eps')(1-\eps) \frac{\card{K}}{k}
    \leq \card{T}
    \leq \frac{1+\eps'}{1-\eps} \cdot \frac{\card{K}}{k}
    \leq (1+2\eps+\eps')\frac{\card{K}}{k}
    \ ,\]
    
    where $\eps+\eps'<1/2$ was used in the last inequality. Thus, for each $v \in K$,
    
    \[
    \card{N(v) \cap T}
    \geq (1-\eps')(1-\eps) \frac{\card{K}}{k}
    \geq \frac{(1-\eps')(1-\eps)^2}{1+\eps'}\card{T}
    = \frac{(1-\eps')^2(1-\eps)^2}{1-\eps'^2}\card{T} 
    \geq (1-2\eps-2\eps')\card{T}
    \ .\qedhere\]
\end{proof}

Note that if $\eps+\eps'<1/4$ the sets $T_i$ defined by an $\eps'$-AC-preserving bucketing within an $\eps$-AC-like set $K$ $2$-hop connect $K$, i.e., $\forall i \in [k], \forall \set{u,v} \subseteq K, \card{N(u)\cap N(v) \cap T_i} \geq (1-4\eps-4\eps')\card{T_i} > 0$.

\begin{lemma}
    \label{lem:rand-ac-preserve}
    Let $\eps,\eps'$ be two positive constants s.t.\ $\eps+\eps'<1/2$. Let $K$ be $\eps$-AC-like and $k$ an integer. Consider a $k$-bucketing of $K$ picked uniformly at random. For each $i \in [k]$, the probability that the $i$th bucket fails to be $\eps'$-AC-preserved is at most $2\card{K}\exp(-\eps'^2 \card{K}/(6k))$.
\end{lemma}
\begin{proof}
    The lemma follows from applying a Chernoff bound (\cref{lem:chernoff}) at each $v\in K$.
\end{proof}

In one of the last steps of our $O(1)$-round \permute algorithm, we perform a second bucketing within previously formed buckets and argue that only a few buckets from this second bucketing are not $\eps''$-AC-preserved.

\begin{Algorithm}
\label{alg:permute-constant}Procedure \permute, in almost-clique $K$, on subset $S \subseteq K$ of the nodes.
\begin{flushleft}
\textbf{Parameters:} Let $C=O(1)$ be a large enough constant, 
\[ \eps':=1/24-\eps,\quad \eps'':=1/12,\quad k := \floor{\Delta/(C\log n)},\quad\text{and} \quad k' := \ceil{C\log \log n} \ . \]
\end{flushleft}

\begin{enumerate}[leftmargin=*]
    \item\label[step]{step:rough-permute-cst} \textbf{Rough bucketing.} Each $v \in K$ independently picks a random $t(v) \in [k]$ u.a.r.

    For each $i\in [k]$, let $T_i:=\set{v \in K : t(v) = i}$ and $S_i := T_i \cap S$.

    \item\label[step]{step:count-rough} \textbf{Counting rough buckets.} For each $i\in [k]$, the nodes in $T_i$ compute and broadcast $\card{T_i}$ and $\card{S_i}$.
    
    \item\label[step]{step:relabel-cst} \textbf{Relabeling.}
    Within each $T_i$, $i\in [k]$, use \relabel on $S_i$.

    \item\label[step]{step:medium-permute-cst} Within each $T_i$, $i\in [k]$,
    \begin{enumerate}
        \item\label[step]{step:second-bucket-cst} \textbf{Fine bucketing.} Each $v\in[T_i]$ picks a random bucket $t'(v)\in[k']$.

        For each $(i,i')\in [k]\times[k']$, let $T_{i,i'}:=\set{v \in T_i : t'(v) = i'}$ and $S_{i,i'} := T_{i,i'} \cap S$.

        \item\label[step]{step:second-bucket-bfs} \textbf{Counting fine buckets.} Compute and broadcast all $\card{T_{i,i'}}$ and $\card{S_{i,i'}}$ for $i' \in [k']$.

        \item\label[step]{step:fine-permute-cst} For each $i' \in [k']$, \textbf{if} $S_{i,i'}$ is $\eps''$-AC-preserved in $T_i$, 

        \begin{description}
            \item[\textbf{then}] \textbf{Permute within fine bucket.} The maximum \ID node of $T_{i,i'}$ aggregates the $O(\log \log n)$-bit labels of $S_{i,i'}$, picks u.a.r.\ a permutation $\rho_{i,i'}$ of $S_{i,i'}$, sends it to $T_{i,i'}$.
            \item[\textbf{else}] each $v \in S_{i,i'}$ joins the set $R$, to be permuted in the next step.
        \end{description}
    \end{enumerate}

    \item\label[step]{step:finish-cst} \textbf{Permuting leftover fine buckets.} 
     \begin{enumerate}        
        \item\label[step]{step:echo-init} Each $v \in R$ picks a random $C \log n$-bit $r(v)$, broadcasts the tuple $(\ID_v,t(v),t'(v),r(v))$.
        
        \item\label[step]{step:echo-broadcast}
        Nodes in $K$ use Many-to-All Broadcast to disseminate the tuples from $R$ to all of $K$.
        
        \item\label[step]{local-perm} For each $(i,i') \in [k]\times [k']$ s.t.\ $S_{i,i'}\subseteq R$, nodes in $S_{i,i'}$ order themselves according to their $r(v)$'s. Let $\rho_{i,i'}$ be the resulting permutation of $S_{i,i'}$.
    \end{enumerate}
    
    \item\label[step]{step:permute-output-cst} \textbf{Output.} $\forall i,i'$, $v \in S_{i,i'}$ takes index $\pi(v) := \rho_{i,i'}(v) + \sum_{j < i} \card{S_j} + \sum_{j' < i'} \card{S_{i,j'}}$ in output.
\end{enumerate}
\end{Algorithm}

\begin{proof}[Proof of \cref{lem:sample-perm-constant}]
    First, our $O(1)$ \permute procedure has an output distribution close to uniform follows from the same argument that showed this property for our $O(\log \log n)$ \permute procedure.

    By \cref{lem:rand-ac-preserve}, \cref{step:rough-permute-cst} (rough bucketing) produces an $\eps'$-AC-preserving bucketing with probability at least $2\card{K}\exp(-\eps'^2\card{K}/(6k)) \leq n^{-\Omega(\eps'^2 C)}$, i.e., w.h.p. We condition on this high-probability event. 
    
    The rough bucketing being $\eps'$-AC-preserving, by \cref{lem:preserving-ac-like}, each $T_i$ is $2(\eps+\eps')$-AC-like, with $2(\eps+\eps') = (1/12)$. Each $T_i$ thus has diameter $2$ and can efficiently count itself and its subset $S_i$ in $O(1)$ rounds during \cref{step:count-rough}. Since every node in $K$ is adjacent to a node in $T_i$, all of $K$ learns all $\card{S_i}$, $i \in [k]$.

    Relabeling works as in the previous $O(\log \log n)$-round algorithm.
    Consider now the second bucketing of \cref{step:second-bucket-cst}. 
    As each $T_{i,i'}$ and $S_{i,i'}$ are of size at most $O(\log n)$, and there are $k' \in O(\log \log n)$ values to count in \cref{step:second-bucket-bfs}, describing all those values only requires $O(\log^2 \log n)$ bits. Counting all of them within $T_i$ by aggregation along a BFS tree can be done in $O(1)$ rounds, and disseminating all values back to $T_i$ is similarly fast.
    
    Within each $T_i$ in which the second bucketing succeed, \cref{step:fine-permute-cst} finishes to permute its elements in $O(1)$ rounds, since the maximum \ID node within each $T_{i,i'}$ only has to send $O(\log n / \log \log n)$ labels of size $O(\log \log n)$ in an $1/3$-AC-like, and thus low diameter, set $T_{i,i'}$.
    
    We finish by arguing that permuting the elements in $R$ within their $S_{i,i'}$ groups can be done in $O(1)$ rounds, w.h.p. By \cref{lem:echo}, if $R$ contains at most $O(\Delta/\log n)$ nodes, Many-to-All-Broadcast succeeds in sharing all of $R$'s tuples in $O(1)$ rounds, w.h.p. The rest of the proof is devoted to showing that $R$ contains $O(\Delta/\log n)$ nodes, w.h.p.
    
    For each $i,i' \in [k] \times [k']$, let $X_{i,i'}$ be the indicator random variable for the $i'$th bucket in $T_i$ not being $\eps'$-AC-preserved. For each $i \in [\card{K}]$, let $Y_i$ be the random variable for the bucket choice of the $i$th node in $K$. Finally, let $f(Y_1,\ldots,Y_{\card{K}}) = \sum_{i =1}^k \sum_{i'=1}^{k'} X_{i,i'}$ be the total number of buckets which are not $\eps'$-AC-preserved.

    From \cref{lem:rand-ac-preserve}, we obtain a bound on each $\Exp[X_{i,i'}]$. Each $T_i$ has size $\card{T_i} \in (1 \pm 1/2) C\log n$, yielding for the aggregate $f$:
    
    \begin{align*}
    \Exp[f]
    = \sum_{i =1}^k \sum_{i'=1}^{k'} \Exp[X_{i,i'}] 
    & \leq k \cdot k' \cdot 4 (C \log n) \cdot e^{-\eps''^2 C \log n/(12k')} \\
    & = 4 \Delta \cdot \log \log n \cdot e^{-C \log n / (12^3 \log \log n)}
    \ .\end{align*}

    For $n$ large enough, or $C$ set to a sufficiently large constant, this yields $\Exp[f] \le \Delta / (2\cdot30^2c^2 \log^4 n)$ where $c:=2C\log n$. Changing the value of each random variable $Y_i$ affects at most two buckets. Therefore $f$ is $2$-Lipschitz. Furthermore, $f$ is $c$-certifiable, as it suffices to reveal the set $T_i$ to show that one of its buckets is not $\eps'$-AC-preserved. Applying Talagrand's inequality (\cref{lem:talagrand}) with a deviation of $t=C\Delta/\log^2 n$, we get that:

    \begin{align*}
    \Pr(f > 4C\Delta/\log^2 n)
    &\leq \Pr\parens*{f > \Exp[f] + t + 30c\sqrt{r \cdot \Exp[f]}} \tag{because $t \ge \Exp[f], 30c\sqrt{2\Exp[f]}$}\\
    &\leq 4 \cdot \exp\parens*{- \frac{t^2}{16 c^2 \Exp[f]}}\\
    &\leq 4 \cdot \exp\parens*{-\frac{C^2\Delta^2/\log^4 n}{16\cdot c^2\cdot \Delta/c^2\log^4 n}} \\
    &\leq 4 \cdot \exp\parens*{- \Delta/16} \ll 1/\poly(n) \tag{since $\Delta\in\Omega(\log^3 n)$}
    \ .\end{align*}

Therefore, with high probability, at most $O(\Delta/\log^2 n)$ buckets join $R$. Each has size $O(\log n)$, so $R$ contains at most $O(\Delta/\log n)$ nodes, w.h.p.
\end{proof}

\section{Coloring in Streaming-Congest}
\label{sec:streaming}

\begin{definition}
\label{def:streaming}
We define \bcstream to be the \bcongest model in which, per round, each node receives the messages from its neighbors in a streaming fashion, using $O(\log^c n)$ memory for some fixed $c > 0$.
\end{definition}

Note that results in \bcstream constrain the size of the messages more than equivalent results in \congest or \bcongest. In the latter models,
the size of the messages can be freely changed between $c \log n$ and $c'\log n$ for two positive constants $c$ and $c'$ without changing $\omega(1)$ asymptotic complexities. This is because, without a memory constraint, for $c > c' > 0$, nodes can simulate an algorithm using $c \log n$-bit messages by buffering the $c' \log n$-bit messages received from each neighbor over $\ceil{c/c'}$ rounds. Such buffering uses $\Theta(\Delta \log n)$ memory and is impossible in \bcstream. In \bcstream, having a $T$-round algorithm for a given problem means that there exist constants $c > 0$ s.t.\ given that nodes can send messages of size $c \log n$, they can solve the problem in $T$ rounds.

Running a randomized color trial remains feasible under \bcstream constraints. As this consists of the core of our algorithm, most steps carry over to this model. The technical difficulties to overcome are: (1) (high-degree) nodes cannot store all colors used in their neighborhood, in order to know their palette; and (2) dense nodes cannot learn the full clique palette nor the full permutation $\pi$ during the synchronized color trial.

Dealing with the first issue is fairly straightforward since in order to overcome the broadcast constraint, nodes sample colors in publicly known sets of colors (e.g., $[\Delta+1]$ or $[x(v)]$). After sampling colors in such a set, a node can learn which sampled colors belong to its palette in one communication round (where each colored node broadcasts its color).

The synchronized color trial (\cref{step:sct} of \cref{alg:high-level}) requires more care. Note that a node $v$ merely needs to know its index in the permutation $\pi(v)$ and the $\pi(v)$-th color in the clique-palette. \Cref{lem:learn-clique-palette,lem:sample-perm-loglog} are both based on the idea of ``random bucketing''. Let us focus on the permutation and consider \cref{alg:permute-loglog}. As each bucket contains $O(\log n)$ nodes, \relabel requires only $\poly\log n$ memory (\cref{alg:relabel}). What remains, then, is to compute the prefix sum $\sum_{j < i} \card{S_j}$ counting the number of elements in buckets of lower indices (\cref{step:permute-output} of \cref{alg:permute-loglog}). Compared to \bcongest, the challenge is to avoid double counting. Indeed, in \cref{step:prefix} of \permute, nodes receive $\Theta(\log n)$ times each term $|S_j|$ of the sum.

Computing prefix sums $\sum_{j < i} |S_j|$ can be done in $O(\log\log n)$ rounds of \bcstream. To achieve this, we progressively merge together the $S_i$'s into larger groups, keeping track of the groups' sizes as they merge. Say groups have size $z$, the main idea is to merge $z^{1/2}$ groups together. Computing the size of the result of this merge involves summing $z^{1/2}$ group sizes. In each group, nodes choose a term to learn in the sum at random (among the $z^{1/2}$ terms). In expectation, $z^{1/2}$ nodes are assigned to each term. Because of the highly connected structure of almost-cliques, we can elect a \emph{unique} node for each term, allowing us to aggregate all values without double counting. Since the sizes of the groups grow polynomially, after $O(\log\log n)$ rounds, all sums have been computed. 

\begin{restatable}{lemma}{prefixsumsLemma}
\label{lem:prefixsums}
Let $T_i$ be a family of sets such as described in \cref{lem:dom-groups}. Suppose nodes of each group $T_i$ knows some value $y_i \le \poly(n)$. There is a $O(\log\log n)$-round \bcstream algorithm such that w.h.p.\ all nodes in $T_i$ learn $\sum_{j < i} y_j$.
\end{restatable}

The same idea allows nodes to find the $i$-th color in the clique palette. When the only remaining nodes are from the put-aside sets, the algorithm only requires $\poly\log n$ memory. Indeed, we can assume the clique palette has size $O(\log^3 n)$ and we sample $O(\log^3 n)$ colors at each step of the process. Observe that the communication procedure described in \cref{lem:echo} works in \bcstream if nodes know in advance which messages they need to store (e.g., the $i$-th color in the clique palette) or if the total number of messages is $\poly\log n$ (e.g., when coloring the put-aside sets).

\subsection{Computing Prefix Sums}
We focus our attention on a clique $K$. We call a \emph{spanning group} a subset $T\subseteq K$ of size $O(\log n)$ and such that for any pair of vertices $u,w\in K$ we have $|T\cap N(u)\cap N(w)| \ge C\log n$. Note that the sets $T_i$ produced by each $v$ sampling a random index $i\in [\Delta/(4C\log n)]$ are a family of disjoint spanning groups, w.h.p.\ (see \cref{lem:dom-groups}).

We begin by dividing the $\set{T_i}_{i\in[k]}$ in ranges of $z_0=C\log n$ groups. Groups in the same range \emph{merge}: they learn their prefix sum \emph{inside the range} as well as the sum of all $y_j$'s in the group. At this point of the algorithm, there are no issues of double counting as each node only learns $O(\log n)$ values determined in advance by its spanning group (\cref{lem:streaming-aggregation-small}).

We then run $O(\log\log n)$ iterations in which we recursively merge groups. At iteration $i$, we merge ranges of $z_i^{1/2}$ groups, where $z_i$ is a lower bound on the size of each group. The size of newly formed groups is at least $z_{i+1} = z_i^{3/2}$. To compute the prefix sums in \cref{lem:prefixsums}, nodes learn $\sum_j y_j$ over groups of smaller index within their range, as well as the sum over all values for its range. Since each range merges $z_i^{1/2} \ll z_i$ groups, we can randomly assign each term of $\sum_j y_j$ to a unique node in each group. Since groups are union of spanning groups $T_i$, they are well connected and allow for simple aggregation. This process is formalized in \cref{lem:streaming-aggregation-large}.

\begin{lemma}
\label{lem:streaming-aggregation-small}
Let $T_1, \ldots, T_k$ be a family of spanning groups and let $z_0=C\log n$. Furthermore, fix some $y_i$ for each $i\in[k]$ and suppose each $v\in T_i$ knows $y_i$. There is a $O(1)$-round \bcstream algorithm such that nodes of $T_i$ learn all $y_j$ for 
$1+\floor*{\frac{i-1}{z_0}}z_0 \le j \le \floor*{\frac{i}{z_0}}z_0$.
\end{lemma}

\begin{proof}
Each node must learn $z_0 < |T_i|$ values.
If each node in $T_i$ broadcasts $y_i$, then a node $v$ can receive (and store) its $z_0=O(\log n)$ values because it has $z_0$ neighbors in each $T_i$. 
\end{proof}

\begin{lemma}
\label{lem:streaming-aggregation-large}
Let $K$ be an almost-clique and $S_1, \ldots, S_m$ be $m$ disjoint subsets of $K$ that are union of disjoint spanning groups, and each of size at least $z \ge C^2\log^2 n$. Furthermore, fix some $y_i$ for each $i\in[m]$ and suppose each $v\in S_i$ know $y_i$. Then, in $O(1)$ rounds of \bcstream, w.h.p.\ nodes of $S_i$ can learn the sums
\begin{itemize}
\item $\sum_j y_j$ where $1+\floor*{\frac{(i-1)}{z^{1/2}}}z^{1/2} \le j < i$, and
\item $\sum_j y_j$ where $1+\floor*{\frac{(i-1)}{z^{1/2}}}z^{1/2} \le j \le \floor*{\frac{i}{z^{1/2}}}z^{1/2}$.
\end{itemize}
\end{lemma}

\begin{proof}
To avoid cumbersome notations, we focus on groups $S_1, \ldots, S_{{z^{1/2}}}$. To prove the lemma, it suffices to repeat the same process in parallel for each contiguous sub-range of $z^{1/2}$ indices in $[m]$.
In each set $S_i$, nodes sample a random value $r(v)\in \range{z^{1/2}}$. We form subsets $R_{i,j}=\set{v\in S_i: r(v)=j}$. In expectation $\Exp[|R_{i,j}|]=|S_i|/z^{1/2} \ge z^{1/2}$ and by Chernoff Bound, w.p. $1-e^{-\Theta(\sqrt{z})}\ge 1-n^{-\Theta(C)}$, all $|R_{i,j}|$ have size at least $z^{1/2}/2$. Furthermore, $R_{i,j}$ has strong diameter 2. Indeed, for any $u,w\in R_{i,j}$, they have $C\log n$ neighbors in each spanning group $T\subset S_i$. Since a spanning group $T\subseteq S_i$ has size $O(\log n)$, there must be at least $z/O(\log n)$ such groups. Counting $C\log n$ shared neighbors in $N(u)\cap N(w)$ for each of the $z/O(\log n)$ spanning group contained in $S_i$, we get that $u$ and $w$ have $\Omega(z)$ common neighbors. Therefore, by Chernoff, w.p. $1-e^{-\Omega(\sqrt{z})}\ge 1-n^{-\Theta(C)}$, $u$ and $w$ have at least $\Omega(z^{1/2})$ common neighbor in $R_{i,j}$.

If nodes $v\in S_i$ broadcast $y_i$ for each $i\in[m]$, because $|S_i| \ge z$ for each $i\in[m]$, we must have $m \le |K|/z \le \Delta/(C\log n)$ different messages. Therefore, they are disseminated in $O(1)$ rounds (by \cref{lem:echo}). Node $v\in R_{i,j}$ for $i,j \le z^{1/2}$ stores only the value $y_j$.

We now explain how to aggregate these values to compute the sums in each $S_i$. Elect an arbitrary \emph{leader} in $S_i$ and arbitrary \emph{chiefs} $R_{i,j}$ for each $j\le z^{1/2}$. Each chief broadcast the \ID of one shared neighbor with the leader. This yields a depth-2 tree, with the leader as root and chiefs as leaves. We aggregate the desired sums on the tree. Note that the chief in group $R_{i,j}$ is the only node in $S_i$ to broadcast $y_j$. This avoids double counting. Once the leader has computed the sums, two rounds of BFS diffuse their values to all nodes in $S_i$.
\end{proof}

\begin{proof}[Proof of \cref{lem:prefixsums}]
We repeatedly aggregate values of larger and larger groups of nodes. We define the following sequence:
\[
    z_0 = C\log n \ , \quad z_1 = z_0^2 \quad\text{and}\quad z_{i+1} = z_i^{3/2} \ .
\]
Our algorithm starts with spanning groups $S_{0,j}:=T_j$ and merges $T_{1+(j-1)z_0}, \ldots, T_{j\cdot z_0}$ together in $S_{1,j}$. It then merges $z_i^{1/2}$ groups $S_{i,1+(j-1)z_i^{1/2}}, \ldots, S_{i,j \cdot z_i^{1/2}}$ into $S_{i+1,j}$ at each iteration. It maintains the invariant $|S_{i,j}|\ge z_i$ for all iterations $i$ and sets $j$. By \cref{lem:streaming-aggregation-small,lem:streaming-aggregation-large}, each iterations takes $O(1)$ rounds. After $O(\log\log\Delta)$ iterations, we merged all groups. Although we describe the process computing $\sum_{j=1}^k y_j$, it is not hard to see that group $T_i$ can also compute the truncated sum $\sum_{j\le i} y_j$.
\end{proof}

\paragraph{Acknowledgements.}
This work was supported by the Icelandic Research Fund grants 217965 and 2310015-051.

\bibliographystyle{alpha}
\bibliography{arxiv-version.bbl}

\newcommand{\etalchar}[1]{$^{#1}$}
\begin{thebibliography}{HKMT21}

\bibitem[AA20]{AA20}
Noga Alon and Sepehr Assadi.
\newblock Palette sparsification beyond ({$\Delta+1$}) vertex coloring.
\newblock In {\em APPROX/RANDOM}, volume 176 of {\em LIPIcs}, pages 6:1--6:22.
  {LZI}, 2020.

\bibitem[ABI86]{alon86}
Noga Alon, László Babai, and Alon Itai.
\newblock A fast and simple randomized parallel algorithm for the maximal
  independent set problem.
\newblock {\em J.\ of Algorithms}, 7(4):567--583, 1986.

\bibitem[ACK19]{ACK19}
Sepehr Assadi, Yu~Chen, and Sanjeev Khanna.
\newblock {Sublinear algorithms for {$(\Delta + 1)$} vertex coloring}.
\newblock In {\em SODA}, pages 767--786. {SIAM}, 2019.

\bibitem[AGM12]{AGM12b}
Kook~Jin Ahn, Sudipto Guha, and Andrew McGregor.
\newblock Graph sketches: sparsification, spanners, and subgraphs.
\newblock In {\em PODS}, pages 5--14. {ACM}, 2012.

\bibitem[AKM22]{AKM22}
Sepehr Assadi, Pankaj Kumar, and Parth Mittal.
\newblock Brooks' theorem in graph streams: a single-pass semi-streaming
  algorithm for {$\Delta$}-coloring.
\newblock In {\em STOC}, pages 234--247. {ACM}, 2022.

\bibitem[AKO20]{assadi2020lower}
Sepehr Assadi, Gillat Kol, and Rotem Oshman.
\newblock Lower bounds for distributed sketching of maximal matchings and
  maximal independent sets.
\newblock In {\em {PODC}}, pages 79--88. {ACM}, 2020.

\bibitem[AKZ22]{AKZ22}
Sepehr Assadi, Gillat Kol, and Zhijun Zhang.
\newblock Rounds vs communication tradeoffs for maximal independent sets.
\newblock In {\em {FOCS}}, pages 1193--1204. {IEEE}, 2022.

\bibitem[Bar16]{Barenboim16}
Leonid Barenboim.
\newblock Deterministic ({\(\Delta+1\)})-coloring in sublinear (in
  {\(\Delta\)}) time in static, dynamic, and faulty networks.
\newblock {\em {J. ACM}}, 63(5):47:1--47:22, 2016.

\bibitem[Bec91]{beck1991algorithmic}
J{\'o}zsef Beck.
\newblock An algorithmic approach to the {L}ov{\'a}sz local lemma. {I}.
\newblock {\em Random Structures \& Algorithms}, 2(4):343--365, 1991.

\bibitem[BEPS16]{BEPSv3}
Leonid Barenboim, Michael Elkin, Seth Pettie, and Johannes Schneider.
\newblock The locality of distributed symmetry breaking.
\newblock {\em {J. ACM}}, 63(3):20:1--20:45, 2016.

\bibitem[BMRT20]{BMRT20}
Florent Becker, Pedro Montealegre, Ivan Rapaport, and Ioan Todinca.
\newblock The impact of locality in the broadcast congested clique model.
\newblock {\em {SIAM Journal on Discrete Mathematics}}, 34(1):682--700, 2020.

\bibitem[CLP20]{CLP20}
Yi-Jun Chang, Wenzheng Li, and Seth Pettie.
\newblock Distributed {($\Delta+1$)}-coloring via ultrafast graph shattering.
\newblock {\em {SIAM Journal on Computing}}, 49(3):497--539, 2020.

\bibitem[CM19]{CM19}
Shiri Chechik and Doron Mukhtar.
\newblock Reachability and shortest paths in the broadcast {CONGEST} model.
\newblock In {\em DISC}, volume 146 of {\em LIPIcs}, pages 11:1--11:13. LZI,
  2019.

\bibitem[DKO14]{drucker2014power}
Andrew Drucker, Fabian Kuhn, and Rotem Oshman.
\newblock On the power of the congested clique model.
\newblock In {\em {PODC}}, pages 367--376. {ACM}, 2014.

\bibitem[Doe20]{Doerr2020}
Benjamin Doerr.
\newblock {\em Probabilistic Tools for the Analysis of Randomized Optimization
  Heuristics}, pages 1--87.
\newblock Springer International Publishing, 2020.

\bibitem[DP09]{DP09}
Devdatt~P. Dubhashi and Alessandro Panconesi.
\newblock {\em Concentration of Measure for the Analysis of Randomized
  Algorithms}.
\newblock Cambridge University Press, 2009.

\bibitem[EPS15]{EPS15}
Michael Elkin, Seth Pettie, and Hsin{-}Hao Su.
\newblock (2{\(\Delta-1\)})-edge-coloring is much easier than maximal matching
  in the distributed setting.
\newblock In {\em SODA}, pages 355--370. {SIAM}, 2015.

\bibitem[FdV22]{FV22}
Sebastian Forster and Tijn de~Vos.
\newblock The laplacian paradigm in the broadcast congested clique.
\newblock In {\em PODC}, pages 335--344. {ACM}, 2022.

\bibitem[FGH{\etalchar{+}}23]{FGHKN22}
Maxime Flin, Mohsen Ghaffari, Magn\'us~M. Halld\'orsson, Fabian Kuhn, and
  Alexandre Nolin.
\newblock A distributed palette sparsification theorem.
\newblock Technical Report 2301.06457, arXiv, 2023.

\bibitem[FHK16]{FHK16}
Pierre Fraigniaud, Marc Heinrich, and Adrian Kosowski.
\newblock Local conflict coloring.
\newblock In {\em FOCS}, pages 625--634. {IEEE} Computer Society, 2016.

\bibitem[FHM23]{FHM23}
Manuela Fischer, Magn\'us~M. Halld\'orsson, and Yannic Maus.
\newblock Fast distributed {B}rooks' theorem.
\newblock In {\em SODA}, pages 2567--2588. {SIAM}, 2023.

\bibitem[GGR21]{GGR20}
Mohsen Ghaffari, Christoph Grunau, and Václav Rozhoň.
\newblock Improved deterministic network decomposition.
\newblock In {\em SODA}, pages 2904--2923, 2021.

\bibitem[GK21]{GK21}
Mohsen Ghaffari and Fabian Kuhn.
\newblock Deterministic distributed vertex coloring: Simpler, faster, and
  without network decomposition.
\newblock In {\em FOCS}, pages 1009--1020. {IEEE} Computer Society, 2021.

\bibitem[HKMT21]{HKMT21}
Magn{\'{u}}s~M. Halld{\'{o}}rsson, Fabian Kuhn, Yannic Maus, and Tigran
  Tonoyan.
\newblock Efficient randomized distributed coloring in {CONGEST}.
\newblock In {\em STOC}, pages 1180--1193. {ACM}, 2021.

\bibitem[HKNT22]{HKNT22}
Magnús~M. Halldórsson, Fabian Kuhn, Alexandre Nolin, and Tigran Tonoyan.
\newblock Near-optimal distributed degree+1 coloring.
\newblock In {\em STOC}, pages 450--463. {ACM}, 2022.

\bibitem[HN23]{HN23}
Magn{\'{u}}s~M. Halld{\'{o}}rsson and Alexandre Nolin.
\newblock Superfast coloring in {CONGEST} via efficient color sampling.
\newblock {\em Theor. Comput. Sci.}, 948:113711, 2023.

\bibitem[HNT22]{HNT22}
Magnús~M. Halldórsson, Alexandre Nolin, and Tigran Tonoyan.
\newblock Overcoming congestion in distributed coloring.
\newblock In {\em PODC}, pages 26--36. {ACM}, 2022.

\bibitem[Hoe63]{Hoeffding}
Wassily Hoeffding.
\newblock Probability inequalities for sums of bounded random variables.
\newblock {\em Journal of the American Statistical Association},
  58(301):13--30, 1963.

\bibitem[HSS18]{HSS18}
David~G. Harris, Johannes Schneider, and Hsin-Hao Su.
\newblock {Distributed {($\Delta + 1$)}-coloring in sublogarithmic rounds}.
\newblock {\em {J. ACM}}, 65:19:1--19:21, 2018.

\bibitem[JN18]{JN18}
Tomasz Jurdzi{\'n}ski and Krzysztof Nowicki.
\newblock Connectivity and minimum cut approximation in the broadcast congested
  clique.
\newblock In {\em SIROCCO}, volume 11085 of {\em LNCS}, pages 331--344.
  Springer, 2018.

\bibitem[Joh99]{johansson99}
{\"{O}}jvind Johansson.
\newblock Simple distributed {$\Delta+1$}-coloring of graphs.
\newblock {\em Inf. Process. Lett.}, 70(5):229--232, 1999.

\bibitem[Lin92]{linial92}
Nathan Linial.
\newblock Locality in distributed graph algorithms.
\newblock {\em {SIAM Journal on Computing}}, 21(1):193--201, 1992.

\bibitem[Lub86]{luby86}
M.~Luby.
\newblock A simple parallel algorithm for the maximal independent set problem.
\newblock {\em {SIAM Journal on Computing}}, 15:1036--1053, 1986.

\bibitem[MT22]{MT20}
Yannic Maus and Tigran Tonoyan.
\newblock Linial for lists.
\newblock {\em Distributed Comput.}, 35(6):533--546, 2022.

\bibitem[Pel00]{peleg00}
David Peleg.
\newblock {\em Distributed Computing: A Locality-Sensitive Approach}.
\newblock SIAM, 2000.

\bibitem[PP19]{PP19}
Shreyas Pai and Sriram~V. Pemmaraju.
\newblock Connectivity lower bounds in broadcast congested clique.
\newblock In {\em PODC}, page 256–258. ACM, 2019.

\bibitem[PS97]{PanconesiS97}
Alessandro Panconesi and Aravind Srinivasan.
\newblock Randomized distributed edge coloring via an extension of the
  {Chernoff-Hoeffding} bounds.
\newblock {\em {SIAM Journal on Computing}}, 26(2):350--368, 1997.

\bibitem[Ree98]{Reed98}
Bruce~A. Reed.
\newblock {\(\omega\)}, {\(\Delta\)}, and {\(\chi\)}.
\newblock {\em J. Graph Theory}, 27(4):177--212, 1998.

\bibitem[RG20]{RG20}
V{\'{a}}clav Rozho\v{n} and Mohsen Ghaffari.
\newblock Polylogarithmic-time deterministic network decomposition and
  distributed derandomization.
\newblock In {\em STOC}, pages 350--363. {ACM}, 2020.

\bibitem[SW10]{SW10}
Johannes Schneider and Roger Wattenhofer.
\newblock A new technique for distributed symmetry breaking.
\newblock In {\em PODC}, pages 257--266. {ACM}, 2010.

\bibitem[Tal95]{Talagrand95}
Michel Talagrand.
\newblock Concentration of measure and isoperimetric inequalities in product
  spaces.
\newblock {\em Publications Math{\'e}matiques de l'Institut des Hautes Etudes
  Scientifiques}, 81(1):73--205, 1995.

\end{thebibliography}
\appendix

\section{Colorful Matching}
\label{sec:appendix-colorful-matching}

Our context differs from \cite{ACK19,FGHKN22} in two ways: some nodes were already colored by the slack generation step and we must reserve a small set of $O(\epsilon \Delta)$ colors (\cref{eq:clique-palette}). This turns out not to be an issue as arguments from \cite{ACK19,FGHKN22} only need enough anti-edges and colors.

Concretely, they define a potential function $\avail_D(F)$ as such. Fix an almost-clique $K$ and a (possibly adversarial) coloring outside $K$. For a set of colors $D$, and some anti-edge $e$ in $K$, define $\avail_D(e)$ the number of colors that anti-edge $e$ can adopt in $D$ (without conflicting with colored neighbors inside or outside $K$). By extension, for a set $F$ of anti-edges, define $\avail_D(F)=\sum_{e\in F} \avail_D(e)$. %
\Cref{thm:colorful-matching} of \cite{FGHKN22} computes a colorful matching as long as enough anti-edges have enough available colors:
\begin{lemma}[Reformulation of \cref{thm:colorful-matching}]
Let $\beta < 1/(18\epsilon)$ be a constant, $D_K \subseteq [\Delta+1]$ and $F_K$ the set of anti-edges in $K$ with both endpoints uncolored. Suppose that for all $K$, we have $\avganti_K \ge C\log n$ and $\avail_{D_K}(F_K) \ge \avganti_K\Delta/3$ for any coloring of $V\setminus K$.
Then, there exists a $O(\beta)$-round algorithm called \matching that computes a colorful matching of size $\beta\cdot\avganti_K$ with probability $1-n^{-\Theta(C)}$ in each almost-clique $K$.
Furthermore, at most $2\beta\cdot\avganti_K$ nodes are colored in each almost-clique during this step.
\end{lemma}

The following lemma shows that almost-cliques with $\avganti_K \ge C\log n$ have a large number of available colors with high probability.

\begin{lemma}
\label{lem:init-avail}
For any almost-clique $K$, let $D=[\Delta+1]\setminus [x(K)]$ and $F$ be the set of anti-edges with both endpoints uncolored after slack generation. With high probability, $\avail_D(F) \ge \avganti_K\Delta^2/3$.
\end{lemma}

\begin{proof}
A node gets colored during slack generation w.p. at most $\ps$. Each time some node $v$ gets colored, $a_v$ anti-edges are removed from $F$. Let $X_v$ be the random variable equal to $a_v$ if $v$ gets colored and zero otherwise. Notice that $X=\sum_{v\in K} X_v$ is an upper bound on the number of edges removed from $F$: for each edge removed, $X$ is charged by at least one of its endpoints. We have $\Exp[X_v] = \ps a_v$ and $X_v \le a_v$. Moreover, $\sum_{v\in K} a_v = \avganti_K|K|$; hence, by a convexity argument, $\sum_{v\in K} a_v^2 \le \frac{\avganti_K|K|}{\epsilon\Delta}\cdot(\epsilon\Delta)^2 = 2\epsilon\avganti_K|K|\Delta$. By Hoeffding inequality \cite[Theorem 2]{Hoeffding},

\begin{align*} 
\Pr\parens{X > 2\Exp[X]} 
    &\le \exp\parens*{-\frac{2\Exp[X]^2}{\sum_{v\in K} a_v^2}} \\
    &\le \exp\parens*{-\frac{2\ps^2\avganti_K^2|K|^2}{2\epsilon\avganti_K|K|\Delta}} \\
    &= \exp(-\Theta(\avganti_K)) = n^{-\Theta(C)}\ .
\end{align*}

Therefore, w.h.p.\ at most $X < 2\Exp[X] = 2\ps\avganti_K|K|$ anti-edges are removed from $F$ by slack generation. This means that $F$ contains at least $(1/2-2\ps)\avganti_K|K| \ge 0.49\avganti_K|K|$ anti-edges. For each edge, at most $2\epsilon\Delta$ colors are blocked from the outside, at most $\Delta/100$ are blocked by nodes in $K$, and at most $10^3\epsilon\Delta$ are blocked by $x(K)$ (see \cref{eq:clique-palette}). Therefore, for $\epsilon\le 10^{-5}$,
\[ \avail_D(F) \ge 0.49\avganti_K|K|\cdot (1-2\epsilon-1/100-10^3\epsilon)\Delta \ge \avganti_K \Delta^2/3 \ .\qedhere\]
\end{proof}

\section{Reducing Put-Aside Sets}
\label{sec:reduce-put-aside}

\begin{Algorithm}
\label{alg:comptry}
Procedure \comptry, in almost-clique $K \in \Kfull$, on uncolored subset $S \subseteq \hK$ of size $O(\Delta / \log n)$.
\begin{flushleft}
\textbf{Parameters:} Let $C=O(1)$ be a large enough constant, $k:= \ceil*{C\log n / \log^2 \log n}$.

Each node $v \in S$ has a publicly known list of colors $L(v)$ of size $\poly(\log n)$ and an $O(\log \log n)$-bit identifier unique within $S$.
For each $v\in S$, let $S^-_v:= \set{u \in S: \ID(u)<\ID(v)}$.
\end{flushleft}
\begin{enumerate}[leftmargin=*]
    \item\label[step]{step:comptry-sample} Each $v \in S$ samples $k$ colors $c_1(v),\ldots, c_k(v)$ in $L(v) \cap \pal{v}$, independently and u.a.r., 
    and disseminates them to $S$ by Many-to-All Broadcast (\cref{lem:echo}).
    \item\label[step]{step:comptry-take} For each $v \in S$, processed in increasing \ID order:

    \textbf{If} $X_v := \set{i \in [k]: c_i(v) \in \pal{v} \setminus \col(S_v^-)} \neq \emptyset$
    \begin{description}
            \item[\textbf{then}] $v$ colors itself with $c_i(v)$, where $i = \min\set{X_v}$. ($\col(v) \gets c_i(v)$)
            \item[\textbf{else}] $v$ stays uncolored. ($\col(v) = \bot$)
        \end{description}
\end{enumerate}
\end{Algorithm}

\reducePutAside*

\begin{proof}
Let us first argue about the bandwidth of \comptry (\cref{alg:comptry}). The procedure has each node in $S$ send $k=\ceil*{\frac{C\log n}{\log^2\log n}}$ colors from publicly known lists of $\poly(\log n)$ colors together its $O(\log \log n)$ \ID. Many-to-All broadcast (\cref{lem:echo}) disseminates these messages to all of $S$ w.h.p.\ in only $O(1)$ rounds, given that $\card{S} \leq O(\Delta / \log n)$. Each disseminated message is of size $O(k \cdot \log \log n + \log \log n) = O(\log n / \log \log n)$, giving the claimed bandwidth.

We now argue the success probability of the procedure. \Cref{alg:comptry} essentially simulates the following sequential algorithm: nodes of $S$, in the order of their {\ID}s, each perform $k$ \trycolor, coloring themselves with the first successful one. They act as if they were connected, never adopting a color already taken by a node of smaller \ID. Colors tried by a node $v$ are sampled independently in a set which does not depend on any other colors tried, so can all be sampled in advance.

This is easily simulated in a distributed setting by \comptry. Once each node $v \in S$ knows the colors and {\ID}s of all other nodes in $S$, it can compute the behavior of all the nodes of smaller \ID $S_v^-$ as they each pick the first of their tried colors not taken by an earlier node, if it exists. Once $v$ has computed the colors adopted by nodes in $S_v^-$, it knows whether it can adopt any of its own colors and potentially color itself.

We now bound the probability that more than $z$ nodes in $S$ fail to get colored. For each $v$, regardless of the colors adopted and tried by the nodes of smaller \ID $S_v^-$, we have 
\[
\card{L(v) \cap \pal{v} \setminus \col(S_v^-)} \geq z\ .
\]
This means that as an uncolored node $v$ tries its $i$th color in the sequential process, regardless of previous \trycolor attempts by $v$ or nodes of smaller \ID, it always succeeds with probability at least
\begin{equation}
\label{eq:prob}
\Pr\parens*{c_i(v) \in L(v)\cap \pal{v} \setminus \col(S^-_v) \mid S^{-}_v} \ge \frac{z}{|L(v)|} \ge \frac{1}{\log\log n \cdot \log^{0.1} n} \ .
\end{equation}

Let $X_v$ be the random variable indicating if $v$ failed to adopt any color by the end of the process. By the chain rule, \cref{eq:prob} implies
\begin{align*} 
\Pr\parens*{X_v=1~|~N^{-}(v)} 
    &= \parens*{1-\frac{1}{\log\log n\cdot\log^{0.1} n}}^k \\
    &\le \exp\parens*{-\frac{C\log^{0.9} n}{\log^3\log n}} \tag{by definition of $k$}\\
    &\le \exp\parens*{-\frac{C\log^{0.1} n}{10}} := p\ .\tag{for $n > 2$}
\end{align*}
The expected number of uncolored nodes is $\Exp[\sum_v X_v] \le p|S| \le z/4$ for large enough $C$. We get concentration by the martingale inequality (\cref{lem:chernoff}). The probability that more than $z$ nodes fail to adopt a color is at most $\exp(-z)$ (by \cref{eq:chernoffmore}). 
\end{proof}

\end{document}